\documentclass[12pt]{article}

\usepackage[T1]{fontenc}
\usepackage[osf]{newtxtext} 
\usepackage{array,cite,enumitem,mathtools,xcolor}
\usepackage[amsthm,nosymbolsc]{newtxmath} 
\usepackage[cal=euler,frak=euler]{mathalpha} 
\usepackage{bm} 

\usepackage[margin=3cm]{geometry} 

\usepackage{titlesec,titling}  
\title{The full electroweak interaction:\\
       an autonomous account}

\author{José M. Gracia-Bondía,$^{1,2}$
Karl-Henning Rehren$^3$
and Joseph C. Várilly$^4$\\[9pt]
{\small
$^1\,$Escuela de Física, Universidad de Costa Rica,
San José 11501, Costa Rica}\\[3pt]
{\small
$^2\,$Centro de Astropartículas y Física de Altas Energías,
Zaragoza 50009, Spain}\\[3pt] 
{\small
$^3\,$Institut für Theoretische Physik, 
Georg-August-Universität Göttingen, 37077 Göttingen, Germany}\\[3pt]
{\small
$^4\,$Escuela de Matemática, Universidad de Costa Rica, 
San José 11501, Costa Rica}
}

\DeclareMathOperator{\ad}{ad}       
\DeclareMathOperator{\T}{T}         
\DeclareMathOperator{\tsum}{{\textstyle\sum}} 

\newcommand{\al}{\alpha}            
\newcommand{\bt}{\beta}             
\newcommand{\dl}{\delta}            
\newcommand{\Dl}{\Delta}            
\newcommand{\eps}{\varepsilon}      
\newcommand{\ga}{\gamma}            
\newcommand{\ka}{\kappa}            
\newcommand{\la}{\lambda}           
\newcommand{\sg}{\sigma}            
\renewcommand{\th}{\theta}          
\newcommand{\ups}{\upsilon}         
\newcommand{\vf}{\varphi}           

\newcommand{\bM}{\mathbb{M}}        
\newcommand{\bS}{\mathbb{S}}        

\newcommand{\g}{\mathfrak{g}}       
\newcommand{\gh}{\mathfrak{h}}      
\newcommand{\gS}{\mathfrak{S}}      
\newcommand{\gu}{\mathfrak{u}}      

\newcommand{\sO}{\mathcal{O}}       

\newcommand{\higgs}{{\mathrm{higgs}}} 
\newcommand{\id}{\mathrm{id}}       
\newcommand{\intl}{{\mathrm{int}}}  
\newcommand{\self}{{\mathrm{self}}} 

\bmdefine{\ee}{e}             
\bmdefine{\VV}{V}             
\bmdefine{\WW}{W}             
\bmdefine{\one}{1}            

\newcommand{\bx}{\boxtimes}         
\newcommand{\del}{\partial}         
\newcommand{\downto}{\downarrow}    
\newcommand{\dsp}{\displaystyle}    
\renewcommand{\geq}{\geqslant}      
\newcommand{\otto}{\leftrightarrow} 
\newcommand{\ox}{\otimes}           
\newcommand{\upto}{\uparrow}        
\newcommand{\w}{\wedge}             
\newcommand{\x}{\times}             
\newcommand{\8}{\bullet}            
\renewcommand{\.}{\cdot}            

\newcommand{\dlc}{\underline{\delta}} 
\newcommand{\dlxx}{\delta_{xx'}}    
\newcommand{\simdel}{\stackrel{\partial}{\sim}}

\newcommand{\half}{{\mathchoice{\thalf}{\thalf}{\shalf}{\shalf}}}
\newcommand{\ihalf}{\tfrac{i}{2}}   
\newcommand{\shalf}{{\scriptstyle\frac{1}{2}}} 
\newcommand{\thalf}{\tfrac{1}{2}}   
\newcommand{\quarter}{\tfrac{1}{4}} 

\newcommand{\mot}[1]{\enspace\text{#1}\enspace} 
\newcommand{\set}[1]{\{\,#1\,\}}  
\newcommand{\word}[1]{\quad\text{#1}\quad} 

\def\wick:#1:{\,\mathopen:#1\mathclose:\,} 

\newcommand{\pd}[2]{\frac{\partial#1}{\partial#2}} 
\newcommand{\pq}[2]{\langle\!\langle#1\,#2\rangle\!\rangle} 
\newcommand{\Tpq}[2]{\langle\!\langle\T#1\,#2\rangle\!\rangle}

\newcommand{\eqIdl}{\stackrel{\mathrm{mod}\,I\delta}{=}}

\theoremstyle{plain}
\newtheorem{thm}{Theorem}[section]    
\newtheorem{prop}[thm]{Proposition}   
\newtheorem{lemma}[thm]{Lemma}        

\theoremstyle{definition}

\theoremstyle{remark}
\newtheorem{remk}[thm]{Remark}        

\numberwithin{equation}{section}      

\newcommand{\hideqed}{\renewcommand{\qed}{}} 

\makeatletter
\titleformat{\section}{\normalfont\large\bfseries}
                      {\thesection}{1em}{}
\titlespacing{\section}{0pt}{*3.5}{*2.3}
\titleformat{\subsection}{\normalfont\normalsize\bfseries}
                         {\thesubsection}{0.7em}{}
\titlespacing{\subsection}{0pt}{*3.25}{*1.5}
\titleformat{\paragraph}[runin]{\normalfont\bfseries}{}{0pt}{}
\titlespacing{\paragraph}{0pt}{\medskipamount}{\wordsep}
\renewcommand{\@dotsep}{200} 
\makeatother

\begin{document}

\maketitle

\thispagestyle{empty}

\begin{flushright}
\itshape
In memory of our dear friend and colleague Florian Scheck
\end{flushright}

\medskip

\begin{abstract}
The precise renormalizable interactions in the bosonic sector of
electroweak theory are intrinsically determined in the autonomous
approach to perturbation theory. This proceeds directly on the
Hilbert--Fock space built on the Wigner unirreps of the physical
particles, with their given masses: those of three massive vector
bosons, a photon, and a massive scalar (the ``higgs''). Neither
``gauge choices'' nor an unobservable ``mechanism of spontaneous
symmetry breaking'' is invoked. Instead, to proceed on Hilbert space
requires using string-localized fields to describe the vector bosons.
In such a framework, the condition of string independence of the
$\bS$-matrix yields consistency constraints on the coupling
coefficients, the essentially unique outcome being the experimentally
known one. The analysis can be largely carried out for other
configurations of massive and massless vector bosons, paving the way
towards consideration of consistent mass patterns beyond those of the
electroweak theory.
\end{abstract}

\begin{flushright}
\begin{minipage}[t]{20pc}
\small
\textit{It is a dereliction of duty for philosophers to repeat the
physicists' slogans rather than asking what is the content of the
reality that lies behind the veil of gauge}
\par \smallskip \hfill
-- John Earman \cite{Earman04}

\smallskip

\textit{The concept of symmetry breaking has been borrowed by the
elementary particle physicists, but their use of the term is strictly
an analogy, whether a deep or a specious one remaining to be
understood}
\par \smallskip \hfill
-- Philip W. Anderson \cite{Anderson72}

\end{minipage}
\end{flushright}

\setlength{\parskip}{1pt} 

\tableofcontents  

\setlength{\parskip}{\smallskipamount} 

%
%

\section{Introduction}
\label{sec:proclaim}

The theory and practice of the autonomous formulation of quantum field
theory \cite{Mund01, MundSY04, MundSY06, Mund07, Mund12, MundDO17,
MundRS17a, MundRS17b, Rosalind, Atropos, MundRS20, Borisov, Gass22a,
Gass22b, MundRS22, MundRS23, Antigone, Athor, RehrenS24}, also called
``string-localized quantum field theory'', or sQFT for short, were
born from dissatisfaction, both with the heuristics permeating the
generally used gauge formalism and with the limitations of algebraic
field theory~\cite{Witten18}. Instead of classical Lagrangians, its
building blocks are the free -- or asymptotic~\cite{Bain00} -- quantum
fields themselves on Fock--Hilbert space, the underlying one-particle
spaces being the irreducible unitary representation spaces of the
Poincaré group as classified in terms of mass and spin or helicity by
Wigner~\cite{Wigner39}. Within extant perturbative approaches to the
phenomenology of particle theory, this undertaking reasonably claims
to be most rigorous, fully enjoying the canonical triad of fundamental
quantum requirements: \textit{positivity} (on which every probability
interpretation hinges), Poincaré \textit{covariance}, and
\textit{locality}, ensuring Einstein causality.

\smallskip

The word ``autonomous'' warrants an explanation. The free fields of
the theory are defined on their physical Hilbert spaces directly,
without ``canonical quantization'' based on classical free
Lagrangians, and without the forced detours through indefinite metrics
and BRST techniques. A consequence is that vector fields associated
with vector bosons necessarily have a weaker localization than usual:
they are localized along some auxiliary ``string'' (whence the name
``sQFT''). One must then impose the \textit{principle of string
independence} (SI), which posits that the $\bS$-matrix must not depend
on the auxiliary string variables. This is necessary and sufficient to
keep consistency with the aforementioned triad of quantum
requirements. The principle turns out to be extremely restrictive on
the allowed interactions, once the field content is specified;
essentially it narrows down the set of admissible interactions to
precisely those found in Nature.

In its practical implementation, there arise several ``obstructions
against string independence'' at each perturbative order, see
Sect.~\ref{ssc:untune-that-string}. The need to cancel all those
obstructions enforces a recursive system of conditions on the
interaction coefficients. Not least, it shows that couplings to a
higgs particle are indispensable in theories with massive vector
bosons~\cite{Aurora}. It turns out that SI holds great power both as a
heuristic device and as a justification tool, dictating
\textit{symmetry from interaction},%
\footnote{Thereby reversing Yang's \textit{dictum}, restated in the
famous terminological discussion on gauge models between Dirac,
Ferrara, Kleinert, Martin, Wigner, Yang himself and Zichichi
\cite{Zichichi84}.}
down to almost every nut and bolt. 

\smallskip

The sQFT method to induce higher interactions by imposing the absence
of obstructions is in fact an offspring of an analogous program
(called ``perturbative'' or ``causal gauge invariance''
\cite{DuetschS99, AsteDS99, Duetsch05, Aurora}, reassembled in the
book \cite{Scharf16}). That program arrives at very similar results by
imposing BRST invariance at all orders. It therefore does not
\textit{start} with the fundamental principle of Hilbert space from
the outset, but imposes the possibility to \textit{recover} a Hilbert
space as its driving mechanism, where sQFT instead imposes
string-independence.

\smallskip

In summary, in the autonomous approach the ``gauge principle'' is
replaced by fundamental quantum principles. This reinforces an early 
objection to regarding gauge invariance as a principle 
\cite{OgietivskiP62}.

\smallskip

The heart of the Standard Model (SM), that is, the fermionic sector of
the electroweak theory in its coupling to the boson sector, was
already investigated in~\cite{Rosalind} by two of us, together with
Jens Mund, on the basis of the sQFT formulation. There we thoroughly 
showed why and how \textit{chirality}%
\footnote{Of the \textit{interaction}, as opposed to some ``intrinsic
nature'' of its carriers.}
is an indispensable trait of flavourdynamics. So-called Yukawa
couplings arise by way of consistency, and not ``in order to give the
leptons a~mass''.

The purpose of the present work is to show that sQFT leads to an
account of the whole electroweak theory from just the knowledge of an
(allowed) particle spectrum of specified masses. One recovers
precisely the phenomenological couplings of flavourdynamics in the~SM,
with massive bosons mediating the weak interactions, and the $\gu(2)$
structure constants -- as for instance in \cite{Scheck12,
Nagashima13}. (One cannot say that we recover the usual
\textit{formulation} of the SM, since our mathematical description of
the boson fields is at variance with gauge theory, and our rule set
does not care for Lagrangians. But the coincidence of the couplings
will be evident.)

In paper~\cite{Rosalind}, the higgs of the~SM was introduced as a
partner for the photon,%
\footnote{Following Okun~\cite{Okun91}, and for obvious grammatical
reasons, henceforth we refer to a (physical) Higgs boson as a higgs,
with a lower-case~h.}
similar to the Stückelberg fields in the Proca-like description of the
massive intermediate vector bosons, or the ``escort fields'' of sQFT
theory itself -- introduced below in Eq.~\eqref{eq:escorts-for-hire}.
Such a partner turned out to be extremely convenient, since the proof
of chirality required its presence. The main goal of the present work
is to complete the tasks in~\cite{Rosalind} and the results on the
Abelian higgs model in~\cite{MundRS23}, by unveiling from first
principles: (a) how (at~least) one quantum scalar particle is
necessarily part of the~SM, and~(b) what the shape of its
self-couplings \textit{must} be, without recourse to alleged,
unobservable~\cite{Lyre08} ``spontaneous symmetry breakings'' and
without pretending that the higgs is ``the giver of mass''.%
\footnote{SSB is not a physical process. Suffice it to say that the
enormous latent heat that would have been released in the early
universe is grossly incompatible with observations
\cite[Sect.~5.C]{Quigg09}.}
Negative-norm states, ghosts, anti-ghosts, are banished as~well.

The fermionic sector and its relation with the bosonic sector having
been dealt with in~\cite{Rosalind}, it remains to analyze the purely
bosonic sector in the present paper. The main task is the exact
determination of all bosonic interactions by consistency at second
order, whereas the higgs self-couplings remain undetermined. The
third-order consistency argument that fixes those self-couplings will
be essentially the same as in the Abelian higgs model \cite{MundRS23}.
One need only make sure that the nonabelian self-interactions do not
interfere with the pertinent conditions.

Our analysis is designed to reach well beyond electroweak theory: we
consider here theories with given numbers of massive and massless
vector bosons -- restricted to only one higgs particle. (The
generalization to more than one higgs is not difficult at second
order, compare~\cite{Aurora}.) Again, the only input is the masses;
all coupling coefficients are determined by string independence of the
$\bS$-matrix at first, second, and third orders in perturbation
theory, compatible with power-counting renormalizability. We are able
to derive all conditions as relations between the masses, the
Yang--Mills-like structure constants of a reductive Lie algebra, and
the higgs couplings and self-couplings. We do not attempt, however, a
general solution of these equations, characterizing all possible mass
and symmetry patterns. In the special case of the electroweak theory,
we give a complete analysis: string independence recovers the
empirically known coefficients, or those otherwise predicted by the
GWS~model.

\paragraph{Plan of the paper.}
After some technical preparations and a short introduction to
string-localized quantum fields describing vector bosons, we solve the
constraints imposed by string independence at first order in the
interactions (Sect.~\ref{sec:IVB}). We then give an outline of
``obstruction theory'', which is the utensil to determine higher order
interactions in sQFT (Sect.~\ref{sec:in-the-way}). We continue in
Sect.~\ref{sec:electro-weak} with the concrete case of the
\textit{electroweak theory}, whose particle content determines the Lie
algebra $\gu(2)$. With the comparison and matching of the sQFT results
with the empirical electroweak interactions, otherwise claimed to be
consequences of gauge invariance, the main goal of the paper is
achieved in Sect.~\ref{sec:electro-weak}.

Some necessary proofs are given in
Sect.~\ref{sec:seconding-the-motion}. The main \textit{structural}
result at second order is found in Sect.~\ref{ssc:soltando-carga}: the
mass-independent part of the self-interactions of massive and massless
vector bosons entails the structure constants of a reductive Lie
algebra; and an induced interaction -- similar to the quartic terms in
gauge theory treatments -- is predicted. Some more conditions on the
coupling coefficients are then identified, and more induced
interactions are found.

Bringing from Sect.~\ref{sec:seconding-the-motion} a few constraints
from second-order string independence, we fix all bosonic couplings
except for the higgs self-interactions, which require a third-order
result. To complete the analysis, a string-independent quartic
self-coupling of the higgs must be admitted, necessary to cancel an
obstruction present at this order (Sect.~\ref{sec:third-order}).
Remarkably, the higgs self-couplings turn out to be ``universal'':
they do not depend on the particulars (masses and Lie algebra) of the
models. Also, by fixing parameters, string independence at third order
shows other putative second-order induced interactions to be absent.

Technical appendices and a brief discussion of models with a more
general particle content are given at the end.

\section{Intermediate vector boson theory}
\label{sec:IVB}

\subsection{Proca and Maxwell field tensors} 
\label{ssc:Proca-and-Maxwell}

Let us start by considering the field strengths $G^{\mu\nu}(x)$ for a
massive boson of spin~$1$, and $F^{\mu\nu}(x)$ for a massless boson of
helicity one (photon). Both are operator-valued distributions on
Hilbert--Fock spaces over the corresponding Wigner's unitary
irreducible representations of the restricted Poincaré group. Their
pertinent time-ordered two-point functions differ only by the mass
parameter:
\begin{align}
\Tpq{G^{\al\mu}(x)}{G_{\bt\rho}(x')}
&= i\bigl( \dl^\al_\bt \del^\mu \del_\rho 
- \dl^\al_\rho \del^\mu \del_\bt - \dl^\mu_\bt \del^\al \del_\rho
+ \dl^\mu_\rho \del^\al\del_\bt \bigr) \,D^F_m(x - x'),
\notag \\
\Tpq{F^{\al\mu}(x)}{F_{\bt\rho}(x')}
&= i\bigl( \dl^\al_\bt \del^\mu \del_\rho
- \dl^\al_\rho \del^\mu \del_\bt - \dl^\mu_\bt \del^\al \del_\rho
+ \dl^\mu_\rho \del^\al \del_\bt \bigr) \,\Dl^F(x - x'),
\label{eq:two-pointer} 
\end{align}
where $D_m^F$ and $\Dl^F \equiv D_0^F$ are respectively the standard
Feynman propagators for massive and massless scalar particles:
\begin{equation}
D_m^F(x) := \frac{1}{(2\pi)^4} \int d^4p\,
\frac{e^{-i(px)}}{p^2 - m^2 + i0},
\word{so that} (\square + m^2) D_m^F(x) = -\dl(x).
\label{eq:Feynman-propagator} 
\end{equation}

The corresponding \textit{potential fields} $B(x)$ and $A(x)$ such
that $G^{\al\mu} = \del^\al B^\mu - \del^\mu B^\al$, as well as
$F^{\al\mu} = \del^\al A^\mu - \del^\mu A^\al$, are however
troublesome -- in wildly different ways. On the one hand, it is well
known that massive bosons of any spin can be described by potential
fields enjoying:
(i)~positivity, i.e., they are operator-valued distributions on the
above indicated Hilbert space;
(ii)~covariance under the Wigner representation on that space; and
(iii)~causal localization.
A paradigmatic example is provided by the above massive (Proca)
particles of spin~$1$. The trouble is that the high-energy behaviour
of such potentials grows steadily worse with spin. Already for~$B$ it
is rather poor, \textit{no better} than that of~$G$:
\begin{align}
\Tpq{B^\al(x)}{B_\bt(x')} =
-\bigl( \dl^\al_\bt + m^{-2} \del^\al \del_\bt \bigr) \,D^F_m(x - x'),
\label{eq:two-pointer-Proca} 
\end{align}
where the derivatives spoil renormalizability in the~UV.

Massless bosons of helicity $|h| \geq 1$, on the other hand, have
free-field descriptions enjoying those desirable properties, too:
think of either the electromagnetic field $F^{\al\mu}$ or the linear
Riemann--Christoffel field $R^{\al\mu\bt\nu}$ for
gravitons~\cite{Antigone}. Nevertheless, the corresponding
\textit{potential} fields do not. The limit as $m \downto 0$ of the
tensor field $G^{\al\mu}(x)$ for the Proca particle \textit{is} of the
Faraday--Maxwell field type; both are positive, covariant and local.
However, whereas the corresponding field potential $B^\mu(x)$ shares
these properties, it clearly possesses no limit as $m \downto 0$ --
and the ordinary electrodynamic potential $A^\mu(x)$ is neither
positive, nor covariant, nor local. Since violation of positivity
conflicts with the probabilistic interpretation of quantum theory, to
salvage positivity \textit{for observables} on use of~$A^\mu(x)$, one
is apparently forced to introduce indefinite metrics, ghost fields,
and the like.

\subsection{Two birds with one stone}
\label{ssc:rara-avis}

The sQFT framework addresses those problems of received QFT -- and
quite a few others. The key point is the definition of \textit{new
field potentials} with desirable properties. These potentials depend
on spatial directions~$e$ (the ``strings'') in Minkowski space, both
for massive carriers of interaction and for massless ones -- like
photons, gluons or gravitons. Their definition is identical in the
massive ($s = 1$) and the massless ($|h| = 1$) cases:
\begin{equation}
A^\mu(x,e) := \int_0^\infty ds\, G^{\mu\nu}(x + se)\,e_\nu  \word{or}
A^\mu(x,e) := \int_0^\infty ds\, F^{\mu\nu}(x + se)\,e_\nu,
\label{eq:A-vector-field} 
\end{equation}
with $e = (e^0,\ee)$ denoting a spacelike direction, taken by
convention in the hyperboloid $H \subset \bM^4$ of unit radius~$1$,
that is, $e_\mu e^\mu = -1$. 

The operations \eqref{eq:A-vector-field} are invertible, namely, there
holds:
\begin{equation}
G^{\mu\nu}(x) = \del^\mu A^\nu(x,e) - \del^\nu A^\mu(x,e),
\word{resp.}
F^{\mu\nu}(x) = \del^\mu A^\nu(x,e) - \del^\nu A^\mu(x,e).
\label{eq:field-recovery} 
\end{equation}

The new potential vector for massless vector bosons (``photons'') has
positive two-point functions on the same Hilbert space as~$F$; it is
covariant, i.e., the string $e$ is Lorentz-transformed under Poincaré
transformations; and localized, i.e., $[A(x,e), A(x',e')] = 0$
whenever the half-lines $\set{x + se : s \geq 0}$ and
$\set{x' + se' : s \geq 0}$ are causally disjoint.

\smallskip

In the massive case, where $G^{\mu\nu}(x)$ is the curl of a pointlike
\textit{Proca field} $B^\mu(x)$, one can write
\begin{equation}
A^\mu(x,e) = B^\mu(x) + \del^\mu \phi(x,e),  \word{where}
\phi(x,e) := \int_0^\infty ds\, B^\mu(x + se)\,e_\mu, 
\label{eq:escorts-for-hire} 
\end{equation}
whereby the scalar ``escort'' field $\phi$ carries away the 
``non-renormalizability'' -- see \cite{MundRS17b}.%
\footnote{One is reminded here of the Stückelberg fields. It is well
known that, in the case of a unique Proca field or ``massive photon'',
perturbative renormalizability of the model can be recovered with
their help~\cite{Burnel86}. However, this fails in the nonabelian
cases \cite{DragonHvN97, DelbourgoTT88, Felicitas}.}
Henceforth we write $A^\mu(x,e)$ for both massive and massless
intermediate vector bosons. 

It is also true (for any value of the mass) that
\begin{equation}
\del_\mu A^\mu(x,e) + m^2 \phi(x,e) = 0; \qquad 
e_\mu A^\mu(x,e) = 0.
\label{eq:waving-away} 
\end{equation}
These relations follow from the definition \eqref{eq:A-vector-field}.
If $m^2 = 0$, \eqref{eq:waving-away} constitutes two constraint
equations for the massless potential; so it has two degrees of freedom
-- as it should. Note that for $A^\mu(x,e)$ the limit $m \downto 0$ is
smooth.

\smallskip

In view of all of the above, and since we do not share the presumption
that interactions brought by massless intermediate vector bosons are
``more natural'' than interactions mediated by massive carriers, we
put massive and massless interaction carriers on the same footing --
as we have done in our investigation of
flavourdynamics~\cite{Rosalind}.

To unify the notation for both massless and massive vector bosons,
from now on we shall write $F^{\mu\nu}$ instead of $G^{\mu\nu}$ in the
massive case also; thus $F^{\mu\nu}(x)$ is the curl of $B^\mu(x)$ when
$m > 0$, and it is the curl of $A^\mu(x,e)$ in both cases.

\subsection{Notations and nomenclature}
\label{ssc:notations}

We use throughout the notation $(VW) \equiv V_\nu W^\nu 
= \eta_{\mu\nu} V^\mu W^\nu = V^0 W^0 - \VV\.\WW$ for Minkowski
products of vectors (including fields or differential operators)
on~$\bM^4$. In particular, $(\del A) \equiv \del_\mu A^\mu$ denotes a
divergence.

\paragraph{String integrations.}
It is convenient to introduce the notation $I^\mu_e$ for integration
in the spacelike direction~$e$, which always appears accompanied by
multiplication with~$e^\mu$; that is, for any field component (or
numerical distribution) $X(x)$ we write
\begin{equation}
I^\mu_e X(x) := e^\mu \int_0^\infty ds\, X(x + se).
\label{eq:line-integral} 
\end{equation}
The formulas \eqref{eq:A-vector-field} may thus be rewritten as
\begin{equation}
A_\mu(x,e) := I^\nu_e F_{\mu\nu}(x),
\label{eq:I-notation} 
\end{equation}
and in the massive case the escort field \eqref{eq:escorts-for-hire}
is given by $\phi(x,e) := I^\nu_e B_\nu(x)$. Assuming that $X(x + se)$
falls off for large~$s$, as is justified (in the sense of correlation
functions) for all pertinent fields, the fundamental theorem of
calculus reverses this integral transformation:
\begin{equation}
\del_\mu(I^\mu_e X)(x) = I^\mu_e(\del_\mu X)(x) = -X(x),
\label{eq:our-TFC} 
\end{equation}
or more briefly, $(\del I_e) = (I_e \del) = -\id$. Indeed,
$$
\del_\mu(I^\mu_e X)(x) = (I^\mu_e \del_\mu X)(x)
= \int_0^\infty \! ds\, e^\mu \del_\mu X(x + se)
= \int_0^\infty \! ds\, \frac{\del}{\del s}\, X(x + se) = - X(x).
$$

The fields $A_\mu(x,e)$ and $\phi(x,e)$ are operator-valued
distributions in both $x$ and~$e$. As advertised, to get rid of
possible singularities in the $e$-dependence, we smear them in~$e$
with an arbitrary sufficiently smooth function $c(e)$, supported in
some small region of the hyperboloid $H$ of spacelike unit vectors:
\begin{equation}
A_\mu(x,c) := \int_H d\sg(e)\, c(e) A_\mu(x,e), \quad
\phi(x,c) := \int_H d\sg(e)\, c(e) \phi(x,e).
\label{eq:Ac-phic} 
\end{equation}
Here $d\sg(e)$ is the Lorentz-invariant measure on~$H$. More
generally, we write
$$
(I^\nu_c X)(x) := \int_H d\sg(e)\, c(e) (I^\nu_e X)(x)
$$
for distributions in~$x$.

For distributions in two variables like $\dl(x - x')$, we write
$I^\nu_c$ and $I'^\nu_c$ according to the string integration in the
variable $x$ or~$x'$. String-integrated delta functions like
$(I^\nu_c \dl)(x - x')$ or $(I'^\nu_c \dl)(x - x')$ will be called 
``string deltas''. If $c$ has weight one with respect to the
integration measure: $\int_H d\sg(e)\, c(e) = 1$, then
\eqref{eq:our-TFC} becomes
\begin{equation}
(\del\,I_c) = (I_c\,\del) = -\id,
\label{eq:Ic-del} 
\end{equation}
and equations \eqref{eq:field-recovery}, \eqref{eq:escorts-for-hire},
and \eqref{eq:waving-away} with the second constraint replaced by
$I_{c,\mu} A^\mu(x,c) = 0$, hold as~well.

\smallskip

The \textit{principle of string independence} -- see
Sect.~\ref{ssc:untune-that-string} -- requires to study string
variations, that is, variations of $c(e)$ by arbitrary functions
$h(e)$ of weight zero. Let us define
$$
\dl_c^h(X(c)) := \frac{d}{dt} X(c + th)\Bigr|_{t=0} \,.
$$
For massive fields, we introduce $u(h) := \dl_c^h(\phi(c))$. The
definition implies that
\begin{equation}
\dl_c^h(A_\mu(x,c)) 
= \dl_c^h\bigl( B_\mu(x) + \del_\mu\phi(x,c) \bigr) = \del_\mu u(x,h).
\label{eq:dc-A-massive} 
\end{equation}

\begin{lemma} 
\label{lm:dc-A}
For the massless field with helicity one, although $\phi(x,c)$ is not
present, there still exists $u(x,h)$ on the photon's Hilbert space
such that
\begin{equation}
\dl_c^h(A_\mu(x,c)) = \del_\mu u(x,h).
\label{eq:dc-A} 
\end{equation}
It is given by $u(x,h) := -I_c^\nu\bigl( \dl_c^h(A_\nu(x,c)) \bigr)$.
\end{lemma}

\begin{proof}
Derivatives and string integrations commute. Thus
\begin{align*} 
\MoveEqLeft{
\del_\mu I_c^\nu\bigl( \dl_c^h(A_\nu(x,c)) \bigr)
+ \dl_c^h(A_\mu(x,c))}
\\
&= \del_\mu I_c^\nu\bigl( \dl_c^h(A_\nu(x,c)) \bigr)
- \del_\nu I_c^\nu\bigl( \dl_c^h(A_\mu(x,c)) \bigr) 
\\
&= I_c^\nu\bigl( \dl_c^h(\del_\mu A_\nu(x,c)
- \del_\nu A_\mu(x,c)) \bigr)
= I_c^\nu\bigl( \dl_c^h(F_{\mu\nu}(x)) \bigr) = 0,
\end{align*}
where \eqref{eq:Ic-del} as well as the smeared version
of~\eqref{eq:field-recovery} have been used.
\end{proof}

\paragraph{Equations of motion.} 
In the sequel we distinguish different vector bosons, massive or
massless, by an index~$a$. For every massive vector boson, there is an
associated \textit{family of fields}:
\begin{equation}
[a] := \bigl\{ 
F_a^{\mu\nu}, A_{a\mu}(c), B_{a\mu}, \phi_a(c), u_a(h) \bigr\}.
\label{eq:happy-family} 
\end{equation}
The associated family of fields for each massless vector boson
(``photon'') is given by
$$
[a] := \{F_a^{\mu\nu}, A_{a\mu}(c), u_a(h)\}.
$$
The family of relevant higgs fields is $[H] := \{H, \del_\mu H\}$. 

\smallskip

All fields within a family satisfy the Klein--Gordon equation with the
respective mass $m_a$ (equal to~$0$ for the photons) or~$m_H$. Within
each massive family~$[a]$, the following equations of motion hold:
\begin{gather}
\del^{[\mu} B_a^{\nu]} = \del^{[\mu} A_a^{\nu]}(c) = F_a^{\mu\nu},
\qquad
\del_\mu F_a^{\mu\nu} = -m_a^2 B_a^\nu,  \qquad
\del_\mu B_a^\mu = 0,
\notag \\
\del_\mu A_a^\mu(c) = -m_a^2 \phi_a(c), \qquad
\del_\mu \phi_a(c) = A_{a\mu}(c) - B_{a\mu} \,.
\label{eq:eom} 
\end{gather}
For the photons the same equations, apart from those involving $B_a$
or $\phi_a(c)$, hold after putting $m_a = 0$. Namely:
$$
\del^{[\mu} A_a^{\nu]}(c) = F_a^{\mu\nu}, \quad
\del_\mu F_a^{\mu\nu} = 0, \word{and} \del_\mu A_a^\mu(c) = 0.
$$

\paragraph{Sectors and types.}
We shall characterize interaction terms by their ``sector'' and their
``type''. The \emph{sector} specifies the families of the fields
making up a Wick product, say $[a][b][H]$ or $[a][b][c]$. The sectors
made from such $[a]$ and~$[H]$ form the bosonic sector of the
electroweak interaction (or generalizations thereof). Those involving
$[H]$ are collectively called the \emph{higgs sector}. The electroweak
fermionic sector involving leptonic currents was studied
in~\cite{Rosalind}.

The \emph{type} specifies the fields in a Wick product, irrespective
of their family, such as $FAA$ or $ABH$. Thus, $\phi_a A_b \del H$
belongs to the sector $[a][b][H]$ and is of type $\phi A \del H$.

In the context of the electroweak theory, we use labels $a = W_1,W_2$
(or simply $1,2$) and $Z,A$ or else $W_+,W_-,Z,A$ for the families of
the massive vector bosons and the photon, according to the standard
terminology. In addition, we shall write
\begin{equation}
W_1 \equiv A_1(c), \enspace W_2 \equiv A_2(c), \enspace
W_\pm \equiv \tfrac{1}{\sqrt{2}} (W_1 \mp iW_2)(c), \enspace
Z \equiv A_Z(c), \enspace  A \equiv A_A(c)
\label{eq:field-basis} 
\end{equation}
for the string-localized vector fields $A_{a\mu}$. Fields of type $A$,
$\phi$, or~$u$ are by default string-dependent, while $F$, $B$, and~$H$
are string-independent.%
\footnote{%
\label{fn:A-vs-A}
There should be no risk of confusion because the gauge potential for
the photon, usually called $A$ in textbooks, is not defined on a
Hilbert space, and thus it simply does not appear in sQFT. Recall
that, whereas the gauge potential raises ``particles'' with four
states per momentum in an indefinite Fock space, our string-localized
photon field $A(c)$ creates precisely the two physical states per
momentum on the Hilbert space.}

\smallskip

\paragraph{On notation.}
We shall drop writing the dependence on $c$ throughout the main body
of the paper. To forestall confusion in the presence of many field
subindices, the notation $\dlc$ (rather than $\dl_c^h$) will be used
for string variations. The notations $X \simdel Y$ and $X \eqIdl Y$
indicate that $X$ and~$Y$ differ by a total derivative:
$X = Y + \del_\mu Z^\mu$, or by a string delta. In expressions with
two or three variables $x,x',x''$, the symbols $\gS_{xx'}$ and
$\gS_{xx'x''}$ denote symmetrization with respect to them; namely,
$\gS_{xx'} f(x,x') := \half\bigl( f(x,x') + f(x',x) \bigr)$, and
analogously for three variables. We omit the notation
$\wick:\text{\,---\,}:$ for Wick products throughout. Operator
products and time-ordered products of two Wick products will be
expanded into Wick products by Wick's theorem, of which we need only
the tree-level part (one contraction) -- see
Section~\ref{sec:in-the-way}.

\subsection{The principle of string independence}
\label{ssc:untune-that-string}

In the sequel, we adopt the rigorous and flexible
Stückelberg--Bogoliubov--Epstein--Glaser formalism of
``renormalization without regularization'' in perturbation theory
\cite{BogoliubovS80, EpsteinG73}. It involves the
\textit{construction} of a unitary scattering operator $\bS[g,c]$
acting on the Fock spaces of the local free fields, functionally
dependent on a multiplet of smooth external fields $g(x)$, with the
stock requisites of causality and Lorentz covariance
\cite[Sect.~3]{Rosalind}. One looks for $\bS[g,c]$ as a time-ordered
exponential series:
\begin{equation}
\bS[g,c] = \T \sum_{k=0}^\infty \frac{i^k}{k!}
\int S_k(x_1,\dots,x_k,c) g(x_1) \cdots g(x_k) \,dx,
\label{eq:S-matrix} 
\end{equation}
In the adiabatic limit $g(x) \upto g$, this is thought of as the
perturbative expansion of the heuristic $\bS$-matrix
\begin{equation}
\T\exp\, i \int dx\, L_\intl(g,c),
\label{eq:practicalS-matrix} 
\end{equation}
where $g$ is a coupling constant, and in sQFT $c$ denotes a common
string smearing function for all fields appearing in $L_\intl$. With
$$
L_\intl = g L_1+ \half g^2 L_2 \,,
$$
the leading term is $g S_1(x,c) = g L_1(x,c)$ -- usually a cubic Wick
polynomial in the free fields, chosen in relation to the physics under
consideration. It specifies a model. The second-order term is of the
form
\begin{equation}
S_2(x,c) = \T[L_1(x,c)\, L_1(x',c)] - i L_2(x,c) \,\dl(x - x').
\label{eq:S2} 
\end{equation}

The ``principle of string independence'' posits that, notwithstanding
the appearance of the string smearing function~$c$ in $\bS[g,c]$, the
adiabatic limit does not depend on~it. In the subsequent sections, it
will become clear how this condition serves to determine the
interactions, order by order. It already constrains the choice of
$L_1$ (Sect.~\ref{ssc:first-things-first}); it determines $L_2$
(Sect.~\ref{sec:seconding-the-motion}); and it must warrant the
absence of higher-order interactions as $\propto g^3$, which would
violate the power-counting bound (Sect.~\ref{sec:third-order}). See
also \cite{MundRS23} for a systematic coverage of the Abelian case,
and \cite{Rehren24a} for an improved and more general formulation at
all orders.

\subsection{The vector boson sector at first order}
\label{ssc:first-things-first}

String independence of the $\bS$-matrix requires that $\dlc(S_n)$ be
a total derivative for every~$n$. Because $S_1(x,c) = L_1(x,c)$ at
first order, this amounts to the condition:
\begin{equation}
\dlc(L_1) = \del_\mu Q_1^\mu
\label{eq:L-Q} 
\end{equation}
for an appropriate vector polynomial $Q_1^\mu(x,c)$ in the fields. For
the self-interactions, we seek $L_{1,\self}$ as a scalar cubic Wick
polynomial in massless and/or massive string-localized vector
potentials $A_a^\mu(c)$, their field strengths $F_a^{\mu\nu}$, and
their escort fields $\phi_a(c)$, which exist only when $m_a > 0$.
Recall that in the latter case
$B_{a\mu} := A_{a\mu}(c) - \del_\mu \phi_a(c)$ is string-independent.

\begin{prop} 
\label{pr:boson-sector}
Apart from the higgs sector, the cubic self-interaction of a
string-local theory of interacting bosons with spin $s = 1$ or
helicity $|h| = 1$ must be of the form
\begin{align}
L_{1,\self}(x,c) &\equiv L^1_{1,\self}(x,c) + L^2_{1,\self}(x,c) 
\label{eq:bosons-mate} 
\\
&= \sum_{abc} f_{abc} F_a^{\mu\nu}(x) A_{b\mu}(x,c) A_{c\nu}(x,c)
+ \sum_{abc} f_{abc} m_{abc}^2 B_a^\mu(x) A_{b\mu}(x,c) \phi_c(x,c),
\notag
\end{align}
where $f_{abc}$ are completely skewsymmetric real coefficients, and
$$
m_{abc}^2 := m_a^2 - m_b^2 + m_c^2 = m_{cba}^2.
$$
Moreover, if particle $b$ is massless, then particles $a$ and~$c$ must
have equal mass:
\begin{equation}
m_b^2 = 0 \mot{and} f_{abc} \neq 0  \implies  m_a = m_c \,.
\label{eq:massless-b} 
\end{equation}
Consequently, if both particles $a$ and $b$ are massless, then
particle~$c$ is massless, too.
\end{prop}

\begin{proof}
(We drop the subscript ``$\self$'' in $L_{1,\self}$ during this
proof.) String independence \eqref{eq:L-Q} requires that $\dlc(L_1)$
be a total derivative $\del_\mu Q_1^\mu$. In the purely massless case,
only the fields $F$ and $A$ are available to build $L_1$, and (by
Lorentz covariance and the power-counting bound) $L_1$ can only be of
type $FAA$ as in $L_1^1$. In this case, complete skewsymmetry of the
coefficients was proved in~\cite{Borisov}.

If massive vector bosons are admitted, one may make a most general
(hermitian, Lorentz-covariant and power-counting renormalizable)
Ansatz:
$$
L_1 = \sum_{abc} f_{abc} F_a^{\mu\nu} A_{b\mu} A_{c\nu}
+ \sum_{abc} g_{abc} B_a^\mu A_{b\mu} \phi_c 
+ \sum_{abc} h_{abc} A_a^\mu A_{b\mu} \phi_c,
$$
with real coefficients $f_{abc}$, $g_{abc}$, $h_{abc}$ satisfying
$f_{abc} = - f_{acb}$ and $h_{abc} = h_{bac}$. Moreover, since $B$ and
$\phi$ fields are necessarily massive, the following supplementary
rules apply:
\begin{align}
g_{abc} &= 0 \quad\text{whenever $a$ or $c$ is massless, and}
\notag \\
h_{abc} &= 0 \quad\text{whenever $c$ is massless}.
\label{eq:suppl} 
\end{align}

We now compute $\dlc(L_1)$ on use of $\dlc(\phi) = u$ and
$\dlc(A_\mu) = \del_\mu u$. To remove terms of type $\del_\mu u$, we
employ
$$
XY\,\del u = \del(XYu) - \del(XY)u
\simdel -(\del X\,Y + X\,\del Y)u,
$$
as well as the equations of motion~\eqref{eq:eom}. Since
\begin{align}
\dlc(F_a^{\mu\nu} A_{b\mu} A_{c\nu})
&= F_a^{\mu\nu} \bigl( \del_\mu u_b A_{c\nu} + A_{b\mu} \del_\nu u_c)
\notag \\
&\simdel m_a^2 B_a^\nu u_b A_{c\nu} - m_a^2 B_a^\mu A_{b\mu} u_c
- \half F_a^{\mu\nu} u_b F_{c\mu\nu} 
- \half F_a^{\mu\nu} F_{b\nu\mu} u_c,
\label{eq:discard} 
\\
\dlc(B_a^\mu A_{b\mu} \phi_c) 
&= B_a^\mu A_{b\mu} u_c + B_a^\mu \del_\mu u_b \phi_c
\simdel B_a^\mu A_{b\mu} u_c - B_a^\mu u_b \del_\mu \phi_c,
\quad\text{and}
\notag \\
\dlc(A_a^\mu A_{b\mu} \phi_c) 
&= \del^\mu u_a A_{b\mu} \phi_c + A_a^\mu \del_\mu u_b \phi_c
+ A_a^\mu A_{b\mu} u_c
\notag \\
&\simdel m_b^2 u_a \phi_b \phi_c - u_a A_{b\mu} \del^\mu \phi_c
+ m_a^2 \phi_a u_b \phi_c - A_a^\mu u_b \del_\mu \phi_c 
+ A_a^\mu A_{b\mu} u_c \,,
\notag
\end{align}
this produces:
\begin{align*}
\dlc(L_1^1)
&\simdel \sum_{abc} f_{abc} \bigl(
m_a^2 u_b (B_a A_c) - m_a^2 u_c (B_a A_b)
- \half u_b (F_a F_c) + \half u_c (F_a F_b) \bigr)
\\
&\quad + \sum_{abc} g_{abc} \bigl( u_c (B_a A_b)
- u_b (B_a  (A_c - B_c)) \bigr)
\\
&\quad + \sum_{abc} h_{abc} \bigl( m_b^2 u_a \phi_b \phi_c
- u_a (A_b (A_c - B_c)) 
\\
&\hspace*{5em}
+ m_a^2 u_b \phi_a \phi_c - u_b (A_a (A_c - B_c)) + u_c (A_a A_b)
\bigr). 
\end{align*}
Next, relabel the indices $abc$ conveniently, using
$f_{abc} = - f_{acb}$ and $h_{abc} = h_{bac}$, to write the sum as
$\sum_b u_b Z_b$ with
\begin{align}
Z_b &= \sum_{ac} 2 m_a^2 f_{abc} (B_a A_c) - f_{abc} (F_a F_c) 
\notag \\
&\quad + \sum_{ac} g_{acb} (B_a A_c) - g_{abc}(B_a (A_c - B_c)) 
\notag \\
&\quad + \sum_{ac} 2 m_a^2 h_{abc} \phi_a \phi_c
- 2 h_{abc} (A_a (A_c - B_c)) + h_{acb} (A_a A_c).
\label{eq:Z} 
\end{align}

String independence requires that $Z_b$ must vanish for all~$b$. In
the sum over $a,c$, let us examine the coefficients of $(F_a F_c)$,
$(A_a A_c)$, $(B_a B_c)$ and $(B_a A_c)$, in turn.
\begin{itemize}[itemsep=0pt]
\item
Firstly, $\sum_{ac} f_{abc} (F_a F_c) = 0$ requires
$f_{abc} + f_{cba} = 0$. Thus, $f_{abc}$ is skewsymmetric under both
$b \otto c$ and $a \otto c$, and hence, it is completely
skewsymmetric. This reproduces the result from the massless case.

\item
For the symmetric coefficients of $(A_a A_c)$, there are two cases to
consider. If all fields are massive, then 
$h_{acb} - 2 h_{abc} + [a \otto c] = 0$ for all~$b$, that is:
$2 h_{acb} - 2 h_{abc} - 2 h_{cba} = 0$, whereby 
$h_{abc} = h_{acb} - h_{cba} = h_{cab} - h_{cba}$. This is both
symmetric and skew under $a \otto b$; hence $h_{\8\8\8} = 0$.

\item
If instead one index, say $b = b'$, refers to a massless field, then
a~priori $h_{acb'} = h_{cab'} = 0$ by \eqref{eq:suppl}. In this case,
the coefficient of $(A_a A_{b'})$ in $Z_c$ equals
$-2h_{acb'} + h_{ab'c} = h_{ab'c}$. Its vanishing, along with that of
$a \otto c$, again implies $h_{\8\8\8} = 0$. Consequently, the third
summation in~\eqref{eq:Z} may be dropped altogether.

\item
The coefficients of $(B_a B_c)$ and of $(B_a A_c)$ in~$Z_b$ now serve
to determine $g_{abc}$. Their vanishing gives
$$
\mathrm{(I)}  \quad g_{abc} + g_{cba} = 0, \word{and}
\mathrm{(II)} \quad g_{abc} - g_{acb} = 2 m_a^2 f_{abc}\,.
$$
When $a,b,c$ are all massive, these relations must hold for all
permutations, and the relations (I), (II) then also imply:
$$
g_{bac} + g_{acb} = -2 m_b^2 f_{abc} \word{and}
- g_{bac} + g_{abc} = 2 m_c^2 f_{abc} \,.
$$
These equations have a unique solution:
\begin{equation}
g_{abc} = (m_a^2 - m_b^2 + m_c^2) f_{abc} =: m_{abc}^2 f_{abc}\,,
\label{eq:all-massive} 
\end{equation}
valid for all permutations.

\item
On the other hand, if say $b'$ is massless, then a~priori only
$g_{ab'c}$ and $g_{cb'a}$ can be nonzero by \eqref{eq:suppl}.
Formula~(I) implies $g_{ab'c} = - g_{cb'a}$; then (II) together with
$a \otto c$ yields:
$$
2 m_a^2 f_{ab'c} = g_{ab'c} = - g_{cb'a} = -2 m_c^2 f_{cb'a}
= 2 m_c^2 f_{ab'c}\,.
$$
This implies $m_a = m_c$ whenever $f_{ab'c} \neq 0$. With this
specification, Eq.~\eqref{eq:all-massive} holds again for all
permutations. In particular, if any two of the fields $a,b,c$ are
massless, then the third one is also massless, and all permutations of
$g_{abc}$ vanish.
\qed
\end{itemize}
\hideqed
\end{proof}

\begin{remk} 
\label{rk:boson-sector}
\begin{enumerate}[itemsep=0pt]
\item 
It follows from~\eqref{eq:massless-b} that $f_{abc}\,m^2_{abc} = 0$
whenever $a$ or~$c$ is massless. This deletes the non-existent terms
``$B_a A_b \phi_c$'' in $L^2_{1,\self}$ whenever $a$ or~$c$ is
massless. Moreover, $L^2_{1,\self}$ contains no terms with more than
one massless index, because in that case $f_{abc}\,m_{abc}^2 = 0$.
However, terms with massless~$b$ may indeed appear.%
\footnote{This qualifies the meaning of the restricted sum $\sum'$ in
\cite[Eq.~(4.1)]{Rosalind}.}

\item 
In Proposition~\ref{pr:Jacobi} we shall show (by SI at second order)
that $f_{abc}$ in fact must satisfy the Jacobi identity, and thus they
are the structure constants of a reductive Lie algebra of compact
type.

\item 
If $a$ and $b$ are massless and $f_{abc} \neq 0$, then $c$ is also
massless. In view of~(ii), this can be reformulated: the structure
constants $f_{abc}$ for the massless particles define a Lie
subalgebra. The latter may be nonabelian, as for instance in QCD.
\end{enumerate}
\end{remk}

We shall also need $Q^\mu_{1\self}$ so that
$\dlc(L_{1 \self}) = \del_\mu Q^\mu_{1\self}$ holds. This can be
obtained, for instance, by collecting the total derivatives discarded
in the first step of the previous proof.

\begin{prop} 
\label{pr:Q1-self}
\begin{equation}
Q^\mu_{1,\self} \equiv Q^{1\mu}_{1,\self} + Q^{2\mu}_{1,\self} 
= 2 \sum_{abc} f_{abc} F_a^{\mu\nu} u_b A_{c\nu}
+ \sum_{abc} f_{abc} m_{abc}^2 B_a^\mu u_b \phi_c \,.
\label{eq:Q1-self} 
\end{equation}
\end{prop}

\begin{proof}
On using $h_{abc} = 0$, we readily retrieve the derivatives discarded
in \eqref{eq:discard}:
\begin{align*}
& \sum_{abc} f_{abc} \bigl( \del_\mu(F_a^{\mu\nu} u_b A_{c\nu})
+ \del_\nu(F_a^{\mu\nu} A_{b\mu} u_c) \bigr)
+ \sum_{abc} g_{abc} \del_\mu(B_a^\mu u_b \phi_c)
\\
&\quad = \del_\mu \sum_{abc} \bigl[ 2f_{abc} F_a^{\mu\nu} u_b A_{c\nu}
+ g_{abc} B_a^\mu u_b \phi_c \bigr] =: \del_\mu Q^\mu_{1,\self}.
\tag*{\qed}
\end{align*}
\hideqed
\end{proof}

\subsection{Let the higgs be with you}
\label{ssc:higgling}

Our $(L_{1,\self}, Q_{1,\self})$ pair above is \textit{not complete},
since bosonic couplings involving massive neutral spinless fields
(i.e., higgses) have not been included, and as we shall see, they play
an important role in our problem. Such physical pointlike scalar
fields do not suffer from the renormalizability issues discussed in
Sect.~\ref{ssc:Proca-and-Maxwell}. As already learned
in~\cite{DuetschS00}, the presence of some higgses, hinted at by the
appearance of escort fields, is \textit{required} for consistency of
models wherever \textit{massive} $A$-fields are present. One can study
what their couplings ought to be from the standpoint of the SI
principle. To simplify matters, we shall suppose that \textit{only
one} higgs field $H(x)$ of mass~$m_H$ is present -- like in the
minimal~SM.

\begin{prop} 
\label{pr:L1-higgs}
The most general $(L,Q)$ pair coupling the vector bosons to the higgs
is of the form
\begin{align}
L_{1,\higgs} &= \sum_{ab} \bigl[
k_{ab} \bigl( A_{a\mu} B_b^\mu H 
+ A_{a\mu} \phi_b \del^\mu H - \half m_H^2 \phi_a \phi_b H \bigr)
+ \ell H^3,
\label{eq:L1-higgs} 
\\
Q^\mu_{1,\higgs} &= \sum_{ab} 
k_{ab} \bigl( u_b B_a^\mu H + u_b \phi_a \del^\mu H \bigr),
\label{eq:Q1-higgs} 
\end{align}
where the coupling matrix $[k_{ab}]$ is symmetric and links 
\emph{only massive fields}.
\end{prop}

\begin{proof}
A general renormalizable cubic Ansatz coupling the higgs to the vector
bosons is of the form
\begin{align*}
L_{1,\higgs}^1 := \sum_{ab} 
k_{ab} A_{a\mu} B_b^\mu H + k'_{ab} A_{a\mu} A_b^\mu H
+ k''_{ab} A_{a\mu} \phi_b \del^\mu H + k'''_{ab} \phi_a \phi_b H 
\bigr],
\end{align*}
with real constants $k^\8_{ab}$ such that $k'_{ab} = k'_{ba}$ and
$k'''_{ab} = k'''_{ba}$. We require that $\dlc(L_{1,\higgs})$ equal
$\del_\mu Q^\mu_{1,\higgs}$, with $Q^\mu_{1,\higgs}$ containing no
$\del u$-terms. On subtracting
$$
\sum_{ab} \del_\mu \bigl[ 
k_{ab} u_a B_b^\mu H + 2 k'_{ab} u_a A_b^\mu H
+ k''_{ab} u_a \phi_b \del^\mu H \bigr]
$$
from the string variation of $L_{1,\higgs}$, one finds that
\begin{align*}
\dlc(L_{1,\higgs}) &= \! \sum_{ab} \bigl[
k_{ab} \del_\mu u_a B_b^\mu H 
+ 2 k'_{ab} \del_\mu u_a A_b^\mu H
+ k''_{ab} \bigl( \del_\mu u_a \phi_b + A_{a\mu} u_b \bigr) \del^\mu H
+ 2 k'''_{ab} u_a \phi_b H \bigr]
\\
&\simdel \! \sum_{ab} \bigl[
k''_{ab} A_{a\mu} u_b \del^\mu H + 2 k'''_{ab} u_a \phi_b H
- k_{ab} u_a B_b^\mu \del_\mu H + 2 k'_{ab} m_b^2 u_a \phi_b H
\\
&\qquad - 2 k'_{ab} u_a A_b^\mu \del_\mu H
- k''_{ab} u_a (A_{b\mu} - B_{b\mu}) \del^\mu H
+ k''_{ab} m_H^2 u_a \phi_b H \bigr].
\end{align*}

The right-hand side (i.e., the last two lines) must vanish, by string 
independence. Comparing coefficients of $u (A \del H)$, one gets
$k''_{ba} - k''_{ab} = 2 k'_{ab}$; this is both symmetric and skew 
under $(a \otto b)$, so $k'_{ab} = 0$ and $k''_{ba} = k''_{ab}$. The 
coefficient of $u (B \del H)$ vanishes only if $k''_{ab} = k_{ab}$,
and that of $u \phi H$ is zero only if 
$2 k'''_{ab} = -k''_{ab} m_H^2 = -k_{ab} m_H^2$.

As a consequence, $k_{ab} = k_{ba}$, too: the matrix $[k_{ab}]$ is 
symmetric. To sum up:
\begin{align*}
L_{1,\higgs}^1 &= \sum_{ab} k_{ab} \bigl( A_{a\mu} B_b^\mu H 
+ A_{a\mu} \phi_b \del^\mu H - \half m_H^2 \phi_a \phi_b H \bigr),
\\
Q^\mu_{1,\higgs} &= \sum_{ab} 
k_{ab} \bigl( u_a B_b^\mu H + u_a \phi_b \del^\mu H \bigr)
= \sum_{ab} k_{ab}\bigl( u_b B_a^\mu H + u_b \phi_a \del^\mu H \bigr).
\end{align*}

We have found it convenient to include in \eqref{eq:L1-higgs} the
unique (up to a real multiple) renormalizable cubic self-interaction
of the higgs, given by $L_{1,\higgs}^2 := \ell H^3$, which trivially
satisfies the $(L,Q)$ condition \eqref{eq:L-Q}.
\end{proof}

The reader may note our assertion that the higgs does not couple to
massless fields, notwithstanding that it was first detected at the LHC
of CERN by its decay into two photons. The answer is of course that
this decay takes place through \textit{loop graphs}. It is well known
that in such cases the one-loop contribution is \textit{finite}. The
(correct) calculation of this contribution is as old
as~\cite{ShifmanVVZ79}. This was questioned in a paper by Gastmans, Wu
and Wu~\cite{GastmansWW15}, which actually received some support from
other calculations. Consensus around the result in~\cite{ShifmanVVZ79}
was reestablished in other papers, among them one by Duch, Dütsch and
one of us~\cite{Aglaea}.

\subsection{Summary: the list of all first-order couplings}
\label{ssc:wrap-up-first}

The complete $L_1$ and $Q^\mu_1$ terms found above are restated as
follows:
\begin{subequations}
\label{eq:L-Q-first-order} 
\begin{align} 
L_1 
&\equiv L_{1,\self}^1 + L_{1,\self}^2 + L_{1,\higgs}^1 + L_{1,\higgs}^2
\notag \\
&= \sum_{abc} f_{abc} F_a^{\mu\nu} A_{b\mu} A_{c\nu}
+ \sum_{abc} f_{abc} m_{abc}^2 B_a^\mu A_{b\mu} \phi_c
\notag \\
&\quad + \sum_{ab} k_{ab} \bigl( B_a^\mu A_{b\mu} H
+ \phi_a A_b^\mu \del_\mu H - \half m_H^2 \phi_a \phi_b H \bigr)
+ \ell H^3;
\label{eq:L-Q-first-order-L1} 
\\
Q^\mu_1 
&\equiv Q_{1,\self}^{1\mu} + Q_{1,\self}^{2\mu} + Q_{1,\higgs}^\mu 
\notag \\
&= 2 \sum_{abc} f_{abc} F_a^{\mu\nu} u_b A_{c\nu} 
+ \sum_{abc} f_{abc} m_{abc}^2 B_a^\mu u_b \phi_c
+ \sum_{ab} k_{ab} \bigl( B_a^\mu u_b H + \phi_a u_b \del^\mu H \bigr).
\label{eq:L-Q-first-order-Q1} 
\end{align}
\end{subequations}
For memory: $f_{abc}$ is completely skewsymmetric, and if $b$ is
massless while $f_{abc} \neq 0$, then $m_a^2 = m_c^2$. In particular,
$f_{abc}\,m_{abc}^2$ vanishes whenever $a,b,c$ stand for one massive
and two massless particles. Moreover, the matrix $[k_{ab}]$ is
symmetric, and it vanishes when either $a$ or~$b$ is massless.

\smallskip

The reader should be aware that until now nothing requires the
coefficients $k_{ab},\ell$ belonging to the higgs sector to be
different from zero. Consideration of the scattering matrix \textit{at
second order} will allow us to establish that in the presence of
massive intermediate bosons these coefficients do not vanish in
general -- and to extract further relations between them. The quartic
interactions $L_2$ will be determined by the condition of string
independence at second order -- see Propositions \ref{pr:four-aces},
\ref{pr:ffkk}, \ref{pr:abcH} and~\ref{pr:aaHH} below. Consideration of
the scattering matrix at third order will show that a term
proportional to~$H^4$ was indeed needed in~$L_2$, and will allow for
the coefficient $\ell$ to be computed.

\section{Obstruction theory}
\label{sec:in-the-way}

At the second order of the $\bS$-matrix, one encounters obstructions
against string independence. These arise because time ordering does
not commute with derivatives:%
\footnote{%
One could also question whether $\dlc$ commutes with time ordering. We
assume this to be the case at tree level (which is all we need), thus
fixing the propagators involving fields $u$ or~$\del u$. With a
broader view, the commutation between $\T$ and~$\dlc$ should be
considered as a renormalization condition on loop contributions.}
in the adiabatic limit,
\begin{equation}
\dlc\Bigl( \int dx\,dx'\, \T\bigl[ L_1(x) L_1(x') \bigr] \Bigr)
= \int dx\,dx'\, \Bigl( \T\bigl[ \del_\mu Q^\mu_1(x) L_1(x') \bigr]
+ [x\otto x'] \Bigr)
\label{eq:delS-at-first} 
\end{equation}
does not vanish because
$\T\bigl[ \del_\mu Q^\mu_1(x) L_1(x') \bigr]
\neq \del_\mu \T\bigl[ Q^\mu_1(x) L_1(x') \bigr]$. Quantities of the
general form 
\begin{equation}
\T\bigl[ \del_\mu Q^\mu(x) L(x') \bigr]
- \del_\mu \T\bigl[ Q^\mu(x) L(x') \bigr]
=: \sO_\mu(Q^\mu(x); L(x')),
\label{eq:OQL} 
\end{equation}
with Wick polynomials $Q^\mu$ and $L$, are the subject of
``obstruction theory'' \cite{Rosalind}, see below. Thus, by 
subtracting $\del_\mu \T[Q^\mu_1(x) L_1(x')] + [x\otto x']$ from the
integrand in \eqref{eq:delS-at-first}, one finds that the total
obstruction at second order (that is, the failure of the integrand
of~\eqref{eq:delS-at-first} to be a derivative) is the symmetric sum
$\sO^{(2)}(x,x')$ displayed below in \eqref{eq:O2}.

As anticipated in~\eqref{eq:S2}, it is decisive that the obstruction
can be cancelled, and string independence can be restored, by adding
``induced'' interactions at second order. However, this requires that
$\sO^{(2)}(x,x')$ be of a form suitable for cancellation: it must be
``resolvable'' -- see Eq.~\eqref{eq:seconds-out}. For this to be
possible, it is instrumental, but of course not sufficient, that
expressions like \eqref{eq:OQL} contain delta functions, which they do
thanks to the subtraction of the derivative term. We aim to show that
the string independence condition (that is: resolvability of
obstructions) imposes constraints on the (until~now) free parameters
of the first-order interactions, and allows to compute the induced
interactions. Through these systematics, sQFT determines interactions.
In this way, in the vein of~\cite{Rosalind} and~\cite{MundRS23}, one
recovers the precise electroweak self-interactions.

In a more ambitious program \cite{Rehren24a}, one can establish that
the coupling terms involving the escort field not only make the
$\bS$-matrix of sQFT string-independent, but make it coincide with the
$\bS$-matrix computed by means of ordinary local fields, like in gauge
theory or the Stückelberg field method of~\cite{Aurora}, once their
unphysical degrees of freedom are eliminated.

We now set out to compute and analyze the structure of obstructions. A
tree-level analysis is sufficient for the purpose of determining
induced interactions. Computing \eqref{eq:OQL} at tree level, Wick's
theorem at the tree level entails:
\begin{equation}
\sO_\mu(Q^\mu(x); L(x')) = \sum_{\psi,\chi'} 
\pd{Q^\mu}{\psi}\, \bigl( \Tpq{\del_\mu\psi(x)}{\chi(x')}
- \del_\mu \Tpq{\psi(x)}{\chi(x')} \bigr) \,\pd{L'}{\chi'} \,.
\label{eq:the-snag} 
\end{equation}
Namely, the uncontracted contributions to the Wick expansion drop out
in the difference, and \eqref{eq:the-snag} are precisely the terms
with one contraction. To evaluate \eqref{eq:O2}, we therefore only
require the \textit{two-point obstructions} of the involved free
fields:
\begin{equation}
\sO_\mu(\psi, \chi')
:= \Tpq{\del_\mu \psi(x)}{\chi(x')} - \del_\mu\Tpq{\psi(x)}{\chi(x')}.
\label{eq:2pt-obst-def} 
\end{equation}

The two-point obstructions are computed from time-ordered two-point
functions (i.e., the ``propagators'') of the free fields and their
derivatives. The latter in turn are determined by ordinary two-point
functions, up to a certain freedom of renormalization. In local QFT,
both terms on the right-hand side of \eqref{eq:2pt-obst-def} coincide
except on the singular set $x = x'$; so that the difference is a
multiple of (a derivative of) $\dlxx \equiv \dl(x - x')$. In the 
string-localized case, these delta functions generalize to string
deltas.

\smallskip

The above is discussed in Appendix~\ref{app:two-pt-obstructions},
where the two-point obstructions are computed; we present here the
results in tabular form. The tables exhibit three real parameters
denoted $c_H$, $c_B$, $c_F$, parametrizing the ambiguity of some
propagators -- see Appendix~\ref{app:two-pt-obstructions}. The freedom
to choose the values of these parameters may be exploited later to
secure string independence.

\begin{table}[htb] 
$$
\begin{array}{@{}l||c|c@{}}
& H' & \del'^\ka H'
\\ \hline
\sO_\mu(H,\cdot) 
& 0 & i c_H \dl_\mu^\ka \dlxx
\\
\sO_\mu(\del^\mu H,\cdot) 
& i\dlxx & -i(1 + c_H) \del^\ka \dlxx
\\ \hline
\end{array}
$$
\caption{Two-point obstructions in the higgs sector}
\label{tbl:higgs-obstructions} 
\end{table}

\begin{table}[htb] 
$$
\begin{array}{@{}l||c|c|c|c@{}}
& F'^{\ka\la} & A'^\ka & B'^\ka & \phi'
\\ \hline
\sO^\mu(F_{\mu\nu},\cdot)
& -i (1 + c_F) \dl_\nu^{[\ka}\del^{\la]} \dlxx \!
& -i (\dl_\nu^\ka - I'_\nu \del^\ka) \dlxx
& -i (1 + c_B) \dl_\nu^\ka \dlxx & -i I'_\nu \dlxx
\\
\sO_{[\mu}(A_{\nu]},\cdot) \!
& -i c_F \dl_\mu^{[\ka}\dl^{\la]}_\nu \dlxx & 0 & 0 & 0 
\\
\sO_\mu(A^\mu,\cdot)
& -i I^{[\ka}\del^{\la]} \dlxx 
& -i (I^\ka - (II')\del^\ka) \dlxx \!
& -i I^\ka \dlxx & -i(II') \dlxx
\\
\sO_\mu(B^\mu,\cdot)
& 0 & 0 & -i(1 + c_B) m^{-2} \del^\ka \dlxx \! & -i m^{-2} \dlxx 
\\
\sO_\mu(\phi,\cdot) & 0 & 0 & -i c_B m^{-2} \dl_\mu^\ka \dlxx & 0 
\\
\sO_\mu(u,\cdot) & 0 & 0 & 0 & 0
\\ \hline
\end{array}
$$
\caption{Two-point obstructions in the Proca sector}
\label{tbl:Proca-obstructions} 
\end{table}

For the massless photon fields ($F,A,u$),
Table~\ref{tbl:Proca-obstructions} holds without the rows and columns
for the non-existent fields $\phi(c)$ and $B = A(c) - \del\phi(c)$.

\subsection{Resolution of obstructions} 
\label{ssc:out-damned-spot}

In order to establish string independence of the $\bS$-matrix, one
must first compute the second-order obstruction
\begin{equation}
\sO^{(2)}(x,x') := \sO_\mu(Q^\mu_1(x); L_1(x')) + [x \otto x'],
\label{eq:O2} 
\end{equation}
with the tree-level formula given in \eqref{eq:the-snag}. Then one
needs to ensure that this obstruction is ``resolvable'', that is, the
result must be of the form~\cite{Athor}:
\begin{equation}
\sO^{(2)}(x,x') = i \dlc(L_2(x)) \,\dl(x - x')
- i \gS_{xx'} \del^x_\mu \, Q_2^\mu(x,x').
\label{eq:seconds-out} 
\end{equation}
If \eqref{eq:seconds-out} can be fulfilled, on then adding to
$(i^2/2)\,g^2 \T[L_1 L'_1]$ the ``induced'' interaction
$(i/2)\,g^2 L_2$, the obstruction of the $\bS$-matrix is cancelled at
second order -- up to a derivative yielding zero in the adiabatic
limit. By \eqref{eq:the-snag}, since $S_1 = L_1$ is cubic in the
fields, $L_2$ will be quartic.

The resolution of obstructions is \emph{the} method to determine
higher-order interactions in sQFT. Not least, condition
\eqref{eq:seconds-out} will require polynomial relations among the
masses, the coefficients $f_{abc}$ in the self-couplings
$L_{1,\self}^1$, $L_{1,\self}^2$, and the $k_{ab}$, $\ell$ in the
higgs couplings $L_{1,\higgs}^1$, $L_{1,\higgs}^2$, see
\eqref{eq:L-Q-first-order-L1}. For the electroweak theory, these
determine all the cubic couplings except~$\ell$, as well as most
quartic terms. String independence at third order will ultimately
determine the coefficients of the cubic and quartic self-couplings of
the higgs field.

\begin{remk} 
\label{rk:no-string-deltas}
The induced interactions $L_2$ are determined by the parts of
\eqref{eq:O2} without string deltas. All obstructions with string
deltas must be total derivatives, contributing a part $Q_2|_{I\dl}$ to
the fields $Q_2^\mu(x,x')$. Lemma~\ref{lm:uu2} in
Appendix~\ref{app:some-lemmata} shows that this is automatically the
case. In order to show that there are no induced third-order
interactions $L_3$ (see Sect.~\ref{sec:third-order}), we need only
third-order obstructions without string deltas, to which $Q_2|_{I\dl}$
does not contribute. We may therefore ignore the string deltas in
Table~\ref{tbl:Proca-obstructions} altogether.
\end{remk}

\subsection{Preliminaries on ``crossings''} 
\label{ssc:crossing-the-bar}

We denote pieces of $\sO_\mu(Q^\mu_1,L_1')$, corresponding to the
pieces of $Q_1$ and $L_1$ in \eqref{eq:L-Q-first-order} as
\textit{crossings}, and write the latter in a somewhat more readable
way as
$$
Q \bx L' := \sO_\mu(Q^\mu; L'),
$$
keeping in mind that according to~\eqref{eq:O2} one must always add
$Q' \bx L$ to this. The full obstruction that has to be resolved is a
sum of all crossings. Each crossing may have contributions belonging
to several sectors. Because cancellations are only possible within
sectors, we shall organize these crossings by field content, that is,
by the sectors and types of the outcome, as in~\cite{Rosalind}.

\section{The electroweak theory}
\label{sec:electro-weak}

As regards the chirality of the Standard Model, we proved in
\cite{Rosalind} that the \textit{physical particle spectrum} with
specified masses \textit{forces} the couplings of the massive bosons
with the fermions to be parity-violating. The proof involved the
crossings of the fermionic interaction set $L^F_1$ with the fermionic
$Q^F_1$ and bosonic $Q_1$ operators. That is, along with the fermionic
obstructions, fermionic-bosonic ones were required there. Now, the
\textit{entire determination} of the electroweak theory from that
spectrum is finally reached in the present paper, by crossing the
bosonic $Q_1$ with the bosonic vertices in~$L_1$. (The crossing of
fermionic $Q^F_1$ with bosonic $L_1$ is inert.)

In Sects.\ \ref{sec:seconding-the-motion} and~\ref{sec:third-order},
we offer a comprehensive computation of the obstructions appearing at
second and third orders in a more general situation than is necessary
in this section: any number of massive vector bosons and photons, plus
a single higgs. String independence is achieved provided several
algebraic relations among the masses, the structure constants
$f_{abc}$, the higgs couplings $k_{ab}$, and the cubic and quartic
higgs self-couplings $\ell$, $\ell'$ are satisfied. (The term
$\ell' H^4$ must be introduced at second order to cancel an upcoming
obstruction at third order.)

In electroweak force theory not all those relations are needed, and
some are redundant. We just anticipate a few pertinent ones from the
general setup, enough to determine the bosonic interactions of the
theory, whose input in sQFT is just \textit{the particle content},
namely the massless photon $A$, three massive vector bosons
$W_1,\,W_2,\,Z$, and the higgs~$H$.

It is assumed that the photon couples to the $W$-bosons, that is,
$f_{12A} \neq 0$. Hence by Eq.~\eqref{eq:massless-b} those have equal
masses $m_W$. It is also assumed that $m_Z \neq m_W$, whereby
$f_{AZ1} = f_{AZ2} = 0$, again by Eq.~\eqref{eq:massless-b}.
(Even so, the limiting case $m_Z = m_W$, also with
$f_{AZ1} = f_{AZ2} = 0$, does yield an acceptable model, but with a
completely decoupled photon.) Thus, $f_{abc}$ are the structure
constants of the Lie algebra $\g = \gu(2)$ in a suitable basis. (The
Jacobi identity does not constrain $f_{12A}$ and $f_{12Z}$.) We shall
also assume that $f_{12Z} \neq 0$. For the possibility $f_{12Z} = 0$,
not realized in Nature, see Remark~\ref{rk:fZ0}.

By a change of basis in the $1$-$2$ mass eigenspace, one may achieve
$k_{12} = 0$; however, the relations $k_{1Z} = k_{2Z} = 0$ are not
assumed for the higgs couplings: that will follow from string
independence. Therefore, the only nontrivial coupling coefficients at
first order are $f_{12A}$ and $f_{12Z}$, the symmetric matrix of higgs
couplings $k_{ab}$ (with $a,b = 1,2,Z$), and the higgs
self-coupling~$\ell$. From the principles of sQFT and that particle
content, we \textit{predict} all the characteristic relations of the
electroweak interactions usually ascribed to gauge theory and
``spontaneous symmetry breaking'' -- none of these constructs have a
place in~sQFT.

We first work out the consequences of the string independence
condition from Proposition~\ref{pr:abcH} further down, to the effect
that
$$
C_{abc} := \sum_{e=1,2,Z}
\frac{k_{ae}}{m_e^2} f_{ebc} m_{ebc}^2 + [a\otto c]
= C_{bca} = C_{cab} \quad (\text{for } a,b,c = 1,2,Z,A).
$$
Since $C_{abc}$ vanishes identically for massless $a$ or~$c$ -- see
Remark~\ref{rk:aurora}(i) -- it must also vanish for massless~$b$.
With $f_{ZAc} = 0$ for all~$c$, the condition $C_{ZA1} = 0$ yields
$k_{Z2} f_{2A1} m_{2A1}^2 = 0$, which implies $k_{Z2} = 0$. Similarly,
$k_{Z1} = 0$. Thus $k_{ab} =: k_a \dl_{ab}$ is diagonal. Then
condition $C_{1A2} = 0$ implies $k_1 = k_2$, and:
\begin{equation}
C_{12Z} = C_{2Z1} = C_{Z12}
\implies \frac{k_1}{m_W^2} = \frac{k_Z}{m_Z^2}
\implies k_a = K m_a^2
\label{eq:k-propto-m} 
\end{equation}
with some constant $K$ (having dimension of inverse mass). In
particular, $k_1 = k_2 =: k_W$. 

\smallskip

With this information, we anticipate and interpret the ``sum rule''
\eqref{eq:sum-rule}, a special case of Eq.~\eqref{eq:ff-kk}. For
$a = 1$ and $b = c = 2$, resp.\ $b = c = Z$, one finds:
\begin{align*}
(f_{12A})^2 \bigl( 4m_W^2 \bigr)
+ (f_{12Z})^2 \bigl( 4m_W^2 - 3m_Z^2 \bigr)
&\stackrel!= k_W^2 = K^2 m_W^4\,,
\\
\shortintertext{and}
(f_{12Z})^2 \biggl[ 
2m_Z^2 + \frac{(m_W^2 - m_Z^2)^2 - m_W^4}{m_W^2} \biggr]
\equiv (f_{12Z})^2 \frac{m_Z^4}{m_W^2} 
&\stackrel!= k_W k_Z = K^2 m_W^2 m_Z^2\,.
\end{align*}
These imply the equalities
\begin{equation}
\frac{(f_{12A})^2}{(f_{12Z})^2} = \frac{m_Z^2 - m_W^2}{m_W^2}
\word{and} f_{12A}^2 + f_{12Z}^2 = K^2 m_W^2\,.
\label{eq:ff-mm} 
\end{equation}
In particular, $m_W < m_Z$ must hold. 

The above results imply that the stronger conditions in
\eqref{eq:ff-kk} are identically satisfied for all $a,b,c,d$. Notice
that $K \neq 0$ unless all couplings are zero; thus, \textit{string
independence imposes the need for the higgs}. The narrative of
``spontaneous symmetry breaking'' is nowhere required. Together with
the higgs self-couplings that will be given by~\eqref{eq:ell-ell'} in
Proposition~\ref{pr:ell-ell'}, the previous relations secure string
independence at all orders.

\smallskip

Let us now compare these results to the Glashow--Weinberg--Salam
model, see~\cite{Schwartz14}. In terms of $\ups := 1/gK$ and
$\la := g^2 K^2 m_H^2/4$, result \eqref{eq:ell-ell'} becomes
\begin{align}
g \ell H^3 + \half g^2 \ell' H^4 = -\half\la (4\ups H^3 + H^4),
\label{eq:higgs-potential} 
\end{align}
which can be written as
$$
\half m_H^2 H^2 - \half\la \bigl( (\ups + H)^2 - \ups^2 \bigr)^2.
$$
In the GWS model, this expression is interpreted as the contribution
of the ``higgs potential'' to the interaction Lagrangian. In sQFT,
such a construct is devoid of meaning (only the interaction part makes
sense). In our treatment, which starts from the particle content of
the theory in order to determine the interactions, it makes sense to
define the Weinberg angle $\Theta$ via the mass ratio:
\begin{align}
\cos\Theta &:= \frac{m_W}{m_Z}\,.
\label{eq:Weinberg-angle} 
\\
\shortintertext{This implies}
\frac{f_{12A}}{f_{12Z}} &= \tan\Theta
\label{eq:ff-Weinberg} 
\end{align}
up to signs that can be absorbed in redefinitions $A \mapsto -A$ or
$Z \mapsto -Z$. 

\goodbreak 

By contrast, in the GWS model the photon field is the gauge field for
the unbroken $U(1)$. That determines the Weinberg angle of the
corresponding rotation in terms of the two gauge coupling constants
for $U(1)$ and $SU(2)$, namely $\tan\Theta := g_1/g_2$. The same
relations \eqref{eq:Weinberg-angle} follow, but here the egg and the
chicken are interchanged. For the sake of comparison, let us identify
the $WWZ$-coupling $g f_{Z12} F_Z^{\mu\nu} W_{1\mu} W_{2\nu}$ of~sQFT
with the corresponding term in the $U(1) \x SU(2)$ Yang--Mills
self-interaction. This amounts to identifying $g f_{12Z}$ with
$-\half g_2 \eps_{123} \cos\Theta$. We may then exploit the freedom of
rescaling $g \mapsto s g$, $K \mapsto s^{-1} K$ to identify $g$ with
$-g_2$, the $SU(2)$ coupling constant.%
\footnote{This choice is in accord with the conventions in
\cite{Rosalind}, producing \cite[Eq.~(4.4)]{Rosalind}, crucially used
in that paper.}
Then the preceding relations imply
$$
f_{12Z} = \half \cos\Theta, \quad
f_{12A} = \half \sin\Theta, \quad
g k_a = g K m_a^2 = \ups^{-1} m_a^2, \quad
m_W^2 = (4K^2)^{-1} = \quarter g_2^2 \ups^2.
$$
This completes the identification between the input parameters of
sQFT: $g$, $m_W = m_Z\cos\Theta$, $m_Z$, $m_H$, with those of the
gauge theory approach: $g_2, g_1 = g_2\tan\Theta, \la, \ups$ in the
bosonic sector of the electroweak theory.

\smallskip

With the standard notations for the particle-antiparticle pair
$W_\pm := (W_1 \mp W_2)/\sqrt{2}$, upon using $B = A + \del\phi$ and
dropping all couplings involving the escort field~$\phi$, there
emerges the total electroweak interaction of sQFT, consisting of:
\begin{itemize}[topsep=3pt, itemsep=0pt]
\item
the cubic self-couplings 
\begin{align*}
& \half g \biggl[ 
\sin\Theta \sum_{abc=1,2,A} \eps_{abc} F_a^{\mu\nu} A_{b\mu} A_{c\nu}
+ \cos\Theta \sum_{abc=1,2,Z} \eps_{abc} F_a^{\mu\nu} A_{b\mu} A_{c\nu}
\biggr]
\\[\jot]
&\quad = \ihalf g_2 \bigl(
F_3^{\mu\nu} W_\mu^+ W_\nu^- + F^{+\mu\nu} W^-_\mu W_{3\nu}
+ F^{-\mu\nu} W_{3\mu} W^+_\nu \bigr),
\end{align*}
where we abbreviate $W_3 := A \sin\Theta + Z \cos\Theta$, a
combination of fields of different masses \cite{Rosalind, Schwartz14},
so that the unit of electric charge%
\footnote{Hence $\Theta = 0$ in the limiting case $m_Z = m_W$, so that
$e = 0$ and the photon decouples.}
is $e = g_2 \sin\Theta$;
\item
the quartic self-couplings 
\begin{align*}
& -\frac{g^2}{4}
\sum_{abcde=1,2,A,Z} f_{abc} f_{ade} (A_b A_d)(A_c A_e)
\\
&= -\half g_2^2 \bigl[ (W^+ W^-)^2 - (W^+ W^+)(W^- W^-) 
+ 2(W^+ W^-)(W_3 W_3) - 2(W^+ W_3)(W^- W_3) \bigr];
\end{align*}
\item
the couplings to the higgs field:
$$
g \bigl[ 2k_W\, W_\mu^+ W^{-\mu} + k_Z Z_\mu Z^\mu \bigr] H
= \ups^{-1}\bigl[ 2m_W^2 W_\mu^+ W^{-\mu} + m_Z^2 Z_\mu Z^\mu \bigr]H;
$$
\item
and finally, the higgs self-couplings \eqref{eq:higgs-potential}.
\end{itemize}
Beyond the displayed terms, there are several couplings involving
escort fields of several types, which have no place in gauge theory --
much as couplings involving ghost fields have no place in~sQFT.
Replacing $A(c)$ by the gauge potential $A$ and ignoring the escort
couplings (whose role is to ensure that the $\bS$-matrix computed with
$A(c)$ is the same as that computed with~$A$), the bosonic interaction
Lagrangian of gauge theory would be recovered.%
\footnote{
\label{fn:cB}
However, a term of type $AAHH$ is missing, 
due to the use of the non-kinematic propagator with $c_B = -1$. The
coupling would be present with $c_B = 0$ via $L_2^*$, see
Appendix~\ref{aps:cH-cB}.}

We recall that the leptonic sector was completed already
in~\cite{Rosalind}. (The only technical difference being the use
in~\cite{Rosalind} of lightlike strings which need no smearing.)

\begin{remk} 
\label{rk:fZ0}
There is another solution with the field content $A$, $W_1$, $W_2$,
$Z$, compatible with all constraints imposed by string independence,
if one were to admit $f_{12Z} = 0$ while $f_{12A} \neq 0$. In physical
terms, the $W$-bosons would be electrically charged, but they would
not couple to the $Z$-boson. The constraints are solved with
$k_{1Z} = k_{2Z} = k_{ZZ} = 0$. Hence the $Z$-boson decouples
completely, and \eqref{eq:ff-Weinberg} does \textit{not} hold. The
resulting admissible theory is a model of a photon and a charged pair
of massive $W$-bosons with a higgs field, tensored with a ``massive
QED'' in which the massive $Z$-boson replaces the photon.
\end{remk}

\section{The boson sector at second order}
\label{sec:seconding-the-motion}

As mentioned, we reorganize the various bosonic crossings sector by
sector (and by~type). The analysis is not restricted to the
electroweak theory only, but is presented with arbitrary given numbers
of massive and massless vector bosons -- limited, however, to only one
higgs particle. The generalization to more than one higgs is not
difficult at second order, cf.~\cite{Aurora}.

\subsection{The Yang--Mills-like sector}
\label{ssc:soltando-carga}

We begin with the main \textit{structural} result at second order: the
Jacobi identity for the completely skewsymmetric coefficients
$f_{abc}$ in \eqref{eq:bosons-mate} is a necessary condition for
string independence at second order. It follows that $f_{abc}$ are
actually the structure constants of a reductive Lie algebra of compact
type (i.e., with negative semidefinite Killing form).

This result follows from the resolution of the obstructions of types
$uFAA$ and $\del uAAA$ in the higgs-free sector. Only the crossings of
$Q^1_{1,\self}$ with $L^1_{1,\self}$ can produce these types. The case
of only massless vector bosons (like QED) was exhaustively examined in
\cite{Borisov} and~\cite{Phocaea}, following~\cite{Gass22a}. One
finds, as~well, a quartic induced interaction $L_2^1$ of type $AAAA$,
familiar from gauge theory. Interestingly enough, this standard
outcome remains valid when massive vector bosons are present; indeed,
the proof is essentially the same as in the massless case.

Let us first look, then, at the types $uFAA$ and $\del uAAA$ in the
higgs-free sector.

\begin{prop} 
\label{pr:Jacobi}
String independence requires that the coefficients $f_{abc}$ of the
Yang--Mills type cubic self-coupling \eqref{eq:bosons-mate} satisfy
the Jacobi identity. Thus, they are the structure constants of a
reductive Lie algebra of compact type, that is, a direct sum of
abelian and simple compact Lie algebras.
\end{prop}

\begin{prop} 
\label{pr:four-aces}
String independence determines the quartic self-coupling to be of the 
form
\begin{equation}
L_2^1 = -2(1 + c_F) \sum_{abcde} f_{abc} f_{ade} (A_b A_d) (A_c A_e).
\label{eq:four-aces} 
\end{equation}
\end{prop}

\begin{proof} We prove both propositions by the same analysis. (The 
notation $L_2^1$ anticipates that the induced interaction $L_2$ is a 
sum of several pieces.)

\paragraph{Step 1: the Jacobi identity.}
Complete skewsymmetry of $\{f_{abc}\}$ is proved in
Proposition~\ref{pr:boson-sector}. To establish the Jacobi identity,
we study the obstruction to string independence arising from
the crossing 
$Q_{1,\self}^1 \bx L_{1,\self}^1 = 2 f_{def} F_d^{\mu\nu} A_{f\nu} u_e
\bx f_{abc} F_a'^{\al\bt} A'_{b\al} A'_{c\bt}$. It contains terms of
type $uFAA$ and $\del uAAA$ that cannot arise from other crossings and
must therefore be separately resolvable. In view of Lemma~\ref{lm:uu2}
we may drop the contribution with a string delta. The crossing at hand
involves three of the obstructions in
Table~\ref{tbl:Proca-obstructions}:
\begin{align}
\sO_\mu\bigl( F^{\mu\nu}, F'^{\al\bt} \bigr) 
&= -i(1 + c_F) \bigl( 
\eta^{\al\nu} \del^\bt - \eta^{\bt\nu} \del^\al \bigr)\,\dlxx,
\label{eq:gravity-sucks} 
\\
\sO_\mu(F^{\mu\nu}, A'_\al) &\eqIdl -i \dl^\nu_\al \,\dlxx,
\notag \\
\sO_{[\mu}(A_{\nu]}, F'^{\al\bt})
&= -i c_F (\dl_\mu^\al \dl_\nu^\bt - \dl_\mu^\bt \dl_\nu^\al) \,\dlxx.
\notag
\end{align}
When pairing $F_d^{\mu\nu} A_{f\nu}$ with $A'_{b\al}$, on writing
$2 F_d^{\mu\nu} A_{f\nu} 
= F_d^{\mu\nu} A_{f\nu} - F_d^{\la\mu} A_{f\la}$, the relevant
obstruction is $\sO_\mu(A_\nu, A'_\al) - \sO_\nu(A_\mu, A'_\al)
= \sO_{[\mu}(A_{\nu]}, A'_\al) = 0$; so we omit this trivial pairing.
Now observe that
\begin{align}
\MoveEqLeft{
Q^1_{1,\self} \bx L_{1,\self}^{1\prime} 
= \sO_\mu(2 f_{def} F_d^{\mu\nu} A_{f\nu}\, u_e;
f_{abc} F_a'^{\al\bt} A'_{b\al} A'_{c\bt})}
\notag \\
&= 2 f_{abc}\, u_e \bigl[
f_{aef} A_{f\nu} \,\sO_\mu(F_d^{\mu\nu}, F_a'^{\al\bt})\,
A'_{b\al} A'_{c\bt}
\notag \\
&\qquad + f_{bef} A_{f\nu} \,\sO_\mu(F_d^{\mu\nu}, A'_{b\al})\,
F_a'^{\al\bt} A'_{c\bt}
+ f_{cef} A_{f\nu} \,\sO_\mu(F_d^{\mu\nu}, A'_{c\bt})\,
F_a'^{\al\bt} A'_{b\al} 
\notag \\
&\qquad + \half f_{dea} F_d^{\mu\nu}
\,\sO_{[\mu}(A_{f\nu]}, F_a'^{\al\bt})\,A'_{b\al} A'_{c\bt} \bigr]
\notag \\
&\eqIdl 2i f_{abc}\, u_e \bigl[ -(1 + c_F) f_{aed} A_{d\nu} \bigl(
A_b'^\nu A'_{c\bt} \, \del^\bt \dlxx
- A'_{b\al} A_c'^\nu \, \del^\al \dlxx \bigr)
\notag \\
&\qquad - f_{bed} A_{d\nu} F_a'^{\nu\bt} A'_{c\bt} \,\dlxx 
- f_{ced} A_{d\nu} F_a'^{\al\nu} A'_{b\al} \,\dlxx 
- c_F f_{dea} F_d^{\mu\nu} A'_{b\mu} A'_{c\nu} \dlxx \bigr].
\label{eq:major-crossing} 
\end{align}
The first line on the right-hand side of \eqref{eq:major-crossing}
consists of two identical terms, summing to:
\begin{align}
& 4i(1 + c_F) f_{abc} f_{aed}\,u_e (A'_c A_d) A'_{b\al} \del^\al \dlxx
\notag \\
&= 4i(1 + c_F) f_{abc} f_{aed} \bigl\{
\del^\al \bigl[ u_e A'_{b\al} (A'_c A_d) \dlxx \bigr]
- \del^\al u_e A'_{b\al} (A'_c A_d) \dlxx \bigr\}
\notag \\
&\quad - 2i(1 + c_F) f_{abc} f_{aed}\,
u_e F_d^{\al\bt} A'_{b\al} A'_{c\bt} \dlxx \,.
\label{eq:minor-crossing} 
\end{align}
To this expression, we add the last line of \eqref{eq:major-crossing}.

\smallskip

To achieve string independence, the final result
of~\eqref{eq:major-crossing} must be the sum of a total derivative and
the string variation of an induced quartic interaction
\cite{MundRS23, Athor}. The first term on the right-hand side of
\eqref{eq:minor-crossing} is a total derivative, and the second one is
a string variation. Indeed:
\begin{align}
f_{abc} f_{aed}\, \dlc \bigl[ (A_e A_b) (A_c A_d) \bigr] 
&= f_{abc} f_{aed}\, \dlc[A_e^\al] A_{b\al} (A_c A_d)
+ f_{aed} f_{abc}\, \dlc[A_b^\al] A_{e\al} (A_d A_c)
\notag \\
&\quad + f_{ade} f_{acb}\, \dlc[A_c^\al] A_{d\al} (A_e A_b)
+ f_{acb} f_{ade}\, \dlc[A_d^\al] A_{c\al} (A_b A_e)
\notag \\
&= 4\,f_{abc} f_{aed}\, \dlc[A_e^\al] A_{b\al} (A_c A_d),
\label{eq:variation} 
\end{align}
using symmetry $(b,c) \otto (e,d)$, and skewsymmetry $(b \otto c)$ and
$(e \otto d)$ to verify that the four summands are all equal. With
$\dlc[A_e^\al] = \del^\al u_e$, this is indeed the second term in
\eqref{eq:minor-crossing} without the factor $-i(1 + c_F) \dlxx$.

\smallskip

It remains to consider the sum of the last term in
\eqref{eq:minor-crossing} and the last line of
\eqref{eq:major-crossing}. Notice the cancellation of the $c_F$-parts.
Dropping the primes in the presence of $\dlxx$, one gets:
\begin{align}
\MoveEqLeft{
-2i f_{abc} u_e \bigl( f_{bed} F_a^{\nu\bt} A_{d\nu} A_{c\bt}
+ f_{ced} F_a^{\al\nu} A_{b\al} A_{d\nu} 
+ f_{aed} F_d^{\al\bt} A_{b\al} A_{c\bt} \bigr) \dlxx}
\notag \\
&= -2i u_e \bigl(
f_{bca} f_{bed} F_a^{\al\bt} A_{d\al} A_{c\bt}
+ f_{cab} f_{ced} F_a^{\al\bt} A_{b\al} A_{d\bt} 
+ f_{abc} f_{aed} F_d^{\al\bt} A_{b\al} A_{c\bt} \bigr) \dlxx
\notag \\
&= -2i u_e \bigl(
f_{acd} f_{aeb} F_d^{\al\bt} A_{b\al} A_{c\bt}
+ f_{adb} f_{aec} F_d^{\al\bt} A_{b\al} A_{c\bt} 
+ f_{abc} f_{aed} F_d^{\al\bt} A_{b\al} A_{c\bt} \bigr) \dlxx
\notag \\
&= 2i \bigl(
f_{abe} f_{acd} + f_{abd} f_{aec} + f_{abc} f_{ade} \bigr)
u_e F_d^{\al\bt} A_{b\al} A_{c\bt} \,\dlxx \,.
\label{eq:Jacobiator} 
\end{align}
Since this expression is neither a derivative not a string variation,
it must vanish in order to achieve string independence. Therefore, the
coefficient in parentheses must vanish -- yielding precisely the
Jacobi identity for the coefficients $f_{abc}$.

\paragraph{Step 2: the Lie algebra.}
The formula $[\xi_a, \xi_b]:= \sum_c f_{abc}\,\xi_c$ defines a Lie
algebra~$\g$. Because the $f_{abc}$ are completely skewsymmetric, the
scalar product $(\xi_a, \xi_b) := \dl_{ab}$ is invariant under the
adjoint representation:
$([\xi, \xi'], \xi'') + (\xi', [\xi, \xi'']) = 0$ for
$\xi,\xi',\xi'' \in \g$. Because this scalar product is nondegenerate,
the adjoint representation is completely reducible; i.e., $\g$~is
reductive. And since the $f_{abc}$ are real, the adjoint representers
$\ad(\xi_b)_{ac} = f_{abc}$ are skewsymmetric, the adjoint Casimir
operator $\sum_a \ad(\xi_a)^2$ is negative semidefinite, thus, $\g$ is
of compact type. Proposition~\ref{pr:Jacobi} is proved.

\paragraph{Step 3: the induced interaction.}
In order to determine the induced interaction $L_2$ such that
$L_\intl = g L_1 + \half g^2 L_2$, one must now add to the string
variation term in the obstruction \eqref{eq:major-crossing} the
(identical) one from the symmetrized crossing $x \otto x'$, and equate
the result with $i \dlc(L_2)\,\dlxx$, see formula
\eqref{eq:seconds-out}. Thus $-2(1 + c_F)$ multiplied by the result
of~\eqref{eq:variation} produces $\dlc(L^1_2)$. Adjusting the
labelling of indices, this yields \eqref{eq:four-aces}, concluding
the proof of Proposition~\ref{pr:four-aces}.
\end{proof}

Note that the derivative terms discarded along the way are still
needed, since they contribute to $Q_2|_\dl$ -- see formula
\eqref{eq:Q2-explicit} below.

\begin{remk} 
\label{rk:massless-Lie-subalgebra}
By the key equation \eqref{eq:massless-b}, the ``massless generators''
$\xi^a$ (with $m_a = 0$) generate a Lie subalgebra~$\gh$, which may be
nonabelian, as for instance in QCD. Also by Eq.~\eqref{eq:massless-b},
their adjoint action on the massive generators assembles the massive
fields into representations of~$\gh$ with constant mass.
\end{remk}

\begin{remk} 
\label{rk:cF}
Expression $L_2^1$ in Proposition~\ref{pr:four-aces} coincides with
the quartic self-interaction typical in gauge theory, with the Feynman
gauge potentials replaced by the string-localized potentials, but with
an extra factor $(1 + c_F)$. After the cancellations in the above
proof, this is the only place where $c_F$ appears at second order, and
-- see the proof of Proposition~\ref{pr:L3} -- it remains undetermined
by string independence also at third order. This situation is
reminiscent of scalar QED, where the renormalization of the propagator
of two derivatives of the scalar field admits a free parameter, which
then changes the coefficient of the quartic interaction, familiar from
the gauge theory treatment. We are therefore free to set $c_F = 0$.
See also footnote~\ref{fn:cB} and Appendix~\ref{aps:cH-cB} for another
instance of the freedom to absorb couplings in renormalizations of
propagators.
\end{remk}

\subsection{Resolution of obstructions in the sectors $[a][b][c][d]$}
\label{ssc:abcd}

From here on, we choose the renormalization constants 
$c_H = -1$, $c_B = -1$; see Appendix~\ref{aps:cH-cB}. We proceed
sector by sector.

The following consistency condition \eqref{eq:ff-kk} provides a link
between the self-coupling constants $f_{abc}$ and the higgs coupling
coefficients $k_{ab}$, and puts constraints on admissible mass
patterns.

\begin{prop} 
\label{pr:ffkk}
String independence in the higgs-free sectors with field content
$[a][b][c][d]$ requires the self-couplings and higgs couplings to
satisfy the relation 
\begin{equation} \sum_e \biggl[
\th(m_e) \biggl( f_{cae} f_{ebd} \frac{m_{cae}^2 m_{ebd}^2}{m_e^2}
- [a \otto b] \biggr) - 2 f_{eab} f_{ced} m_{ced}^2 \biggr]
= k_{ac} k_{bd} - k_{ad} k_{bc}
\label{eq:ff-kk} 
\end{equation}
for all $a,b$ and all massive $c,d$. 

Then all obstructions in the sectors with field content $[a][b][c][d]$
can be resolved. As well as $L_2^1$ in Proposition~\ref{pr:four-aces},
there is a second higgs-free induced interaction:
\begin{equation}
L_2^2 = -\frac{1}{4}\,\mu_H^2
\Bigl( \sum_{ab} k_{ab} \phi_a \phi_b \Bigr)^2 \,.
\label{eq:L2-2} 
\end{equation}
\end{prop}

\begin{remk} 
\label{rk:aurora}
\begin{enumerate}[itemsep=0pt]
\item 
The left-hand side of \eqref{eq:ff-kk} is skewsymmetric under both
$a \otto b$ and $c \otto d$. Moreover, if $c$ or~$d$ is massless, then
$f_{cae} m_{cae}^2 = f_{ebd} m_{ebd}^2 = f_{ced} m_{ced}^2 = 0$
from~\eqref{eq:massless-b} -- consult Proposition
\ref{pr:boson-sector} -- and in this case the expression identically
vanishes.

\item 
Equation~\eqref{eq:ff-kk} is equivalent \textit{mutatis mutandis} to
\cite[Eq.~(20)]{Aurora}, which was derived in a different setting in
terms of Stückelberg fields and BRST invariance, rather than string
independence.

\item 
A useful special case of~\eqref{eq:ff-kk} is $d = a$, which gives the
``sum rule'':
\begin{equation}
\sum_e f_{aeb} f_{aec} \biggl[ m_{aeb}^2 + m_{aec}^2
+ \th(m_e)\frac{(m_a^2 - m_b^2)(m_a^2 - m_c^2) - m_e^4}{m_e^2} \biggr]
= k_{aa} k_{bc} - k_{ab} k_{ac}.
\label{eq:sum-rule} 
\end{equation}
This sum rule is trivially satisfied for massless $a$, $b$ or~$c$.

\item 
The sum rule shows that the higgs couplings are \textit{indispensable}
for string independence in theories with nonabelian massive vector
bosons. To wit, one can always find labels $a,b,c$ for which the
left-hand side is nonzero. For instance, when all masses are equal,
the sum over all $b$ on the left-hand side gives $m^2$ times the
quadratic Casimir operator in the adjoint representation, which is not
zero. Hence the matrix $[k_{ab}]$ cannot vanish.
\end{enumerate}
\end{remk}

\begin{proof}[Proof of Proposition~\ref{pr:ffkk}]
By Lemma \ref{lm:uu2}, we may discard all obstructions with string
deltas. Tables \ref{tbl:higgs-obstructions}
and~\ref{tbl:Proca-obstructions} show that all obstructions without
string deltas in the sectors $[a][b][c][d]$ arise from
\begin{equation}
Q_{1,\higgs} \bx L_{1,\higgs}^{1\prime} + Q_{1,\self} \bx L'_{1,\self}
+ [x \otto x'].
\label{eq:HHss} 
\end{equation}

\paragraph{Step 1.}
Higgs-free obstruction arising from higgs-higgs crossings: of the
first summand in Eq.~\eqref{eq:HHss} obviously we need only the terms
arising through pairings of the higgs fields. We use
Table~\ref{tbl:higgs-obstructions}. With $c_H = -1$, this gives
\begin{align*}
& \biggl[ \pd{Q_{1,\higgs}}{(\del H)} \pd{L_{1,\higgs}^{1\prime}}{H}
- \pd{Q_{1,\higgs}}{H} \pd{L_{1,\higgs}^{1\prime}}{(\del H)} \biggr]
\,i\dlxx
+ [x\otto x']
\\[\jot]
&\quad = 2 \sum_{abcd} \bigl[ k_{ad} k_{bc}
u_a \phi_d \bigl( (A_b B_c) - \half m_H^2 \phi_b \phi_c \bigr)
- k_{ac} k_{bd} u_a (B_c A_b) \phi_d \bigr] \,i\dlxx,
\end{align*}
so that one obtains
\begin{equation}
Q_{1,\higgs} \bx L'_{1,\higgs} + [x \otto x']
= \dlc(L^2_2) \,i\dl_{xx'} + \sO_\higgs \,i\dlxx,
\label{eq:O-higgs} 
\end{equation}
with $L_2^2 = -\quarter m_H^2 \sum_{abcd} 
k_{ab} k_{cd} \,\phi_a \phi_b \phi_c \phi_d$, and the non-resolvable
obstruction is of the form
\begin{equation}
\sO_\higgs
:= 2 \sum_{abcd} (k_{ad} k_{bc} - k_{ac} k_{bd})
u_a (A_b B_c) \phi_d \,.
\label{eq:O-higgs-non} 
\end{equation}
Notice the manifest skewsymmetry of the coefficients in $a \otto b$
and in $c \otto d$.

\paragraph{Step 2.} 
Higgs-free obstruction arising from self-self crossings: we compute
the second summand in Eq.~\eqref{eq:HHss}, using
Table~\ref{tbl:Proca-obstructions}. In computations of this type it is
important to keep Eq.~\eqref{eq:massless-b} in mind. The vanishing
$m_{abe}^2 = m_{eba}^2 = 0$ when $e$~is massless eliminates the
corresponding coupling term, hence it always overrules the occurrence
of factors $m_e^{-2}$ coming from pairings of such terms, as in
Eq.~\eqref{eq:O-self-non}. We multiply such terms by
$\th_e \equiv \th(m_e)$.

One may drop all two-point obstructions involving string deltas, again
by Lemma~\ref{lm:uu2}. The crossing $FuA \bx F'A'A'$ has been dealt
with in Proposition~\ref{pr:four-aces}, and there is a null crossing:
$Bu\phi \bx F'A'A' = 0$. We compute the remaining crossings in
$Q^\mu_{1,\self}\bx L_{1,\self}^{2\prime}$. With $c_B = -1$, one gets:
\begin{align}
& \sum_e \biggl[ 
\pd{Q_{1,\self}}{\phi_e} m_e^{-2} \pd{L_{1,\self}^{2\prime}}{B_e} 
- \pd{Q_{1,\self}}{B_e} m_e^{-2} \pd{L_{1,\self}^{2\prime}}{\phi_e}
- \pd{Q_{1,\self}}{F_e}\,\pd{L_{1,\self}^{2\prime}}{A_e}
\biggr] \,i\dlxx + [x\otto x']
\notag \\
&= 2 \biggl[ \sum_{abcde} \th_c \th_d \th_e \bigl[
f_{cae} m_{cae}^2 f_{ebd} m_{ebd}^2 m_e^{-2} u_a (B_c A_b) \phi_d
- f_{ead} m_{ead}^2 f_{cbe} m_{cbe}^2 m_e^{-2} u_a \phi_d (B_c A_b)
\bigr]
\notag \\
&\qquad -2 \sum_{abcde} 
\th_c \th_d f_{eab} f_{ced} m_{ced}^2 u_a (A_b B_c) \phi_d \biggr]
\,i\dlxx
\label{eq:O-self-non} 
\\
&= 2 \sum_{abcd} \sum_e \biggl[ \th_e \Bigl[
f_{cae} f_{ebd} \frac{m_{cae}^2 m_{ebd}^2}{m_e^2} - [a \otto b] \Bigr]
- 2 f_{eab} f_{ced} m_{ced}^2 \biggr] u_a (A_b B_c) \phi_d \,i\dlxx.
\notag 
\end{align}
We notice that this term is of the same type $uAB\phi$ as
\eqref{eq:O-higgs-non}, with the same skewsymmetry of the coefficients
in $a \otto b$ and $c\otto d$. In view of~\eqref{eq:ff-kk}, it cancels
\eqref{eq:O-higgs-non}.
\end{proof}

\subsection{Resolution of obstructions in sectors $[a][b][c][H]$}
\label{ssc:higgs-odd}

String independence at second order admits an induced interaction of
type $\phi \phi \phi H$, that will be eliminated later at third order
by means of Proposition~\ref{pr:ell-ell'}.

\begin{prop} 
\label{pr:abcH}
The obstruction in the sectors with field content $[a][b][c][H]$
equals:
\begin{align}
& Q_{1,\self} \bx L'_{1,\higgs} + Q_{1,\higgs} \bx L'_{1,\self}
+ [x\otto x']
\notag \\
&\quad = \sum_{abc} \bigl[ C_{abc} m_H^2 \phi_a u_b \phi_c H
- 2(C_{abc} - C_{acb}) (B_a u_b A_c H + \phi_a u_b A_c \del H) \bigr]
\,i\dlxx,
\label{eq:O-abcH} 
\end{align}
where
\begin{equation}
C_{abc} := \sum_e \biggl[ \frac{k_{ae}}{m_e^2} f_{ebc} m_{ebc}^2 
+ \frac{k_{ce}}{m_e^2} f_{eba} m_{eba}^2 \biggr] = C_{cba} \,.
\label{eq:C-abc} 
\end{equation}
String independence requires that $C_{abc}$ be completely symmetric in
$a,b,c$:
\begin{equation}
C_{abc} = C_{acb} = C_{bac}.
\label{eq:abcH} 
\end{equation}
Then the obstruction is resolved by the induced interaction:
\begin{equation}
L_2^5 := \frac{1}{3} \sum_{abc} C_{abc} m_H^2 \phi_a \phi_b \phi_c H.
\label{eq:L2-5} 
\end{equation}
\end{prop}

\begin{proof}
Refer again to the $(L_1,Q_1)$ pair listed
in~\eqref{eq:L-Q-first-order}. The relevant obstructions in sectors
$[a][b][c][H]$ arise only in the crossings:
\begin{align}
& \sum_{abc} f_{abc} \bigl[ 2 F_a^{\mu\nu} u_b A_{c\nu}
+ m_{abc}^2 B_a^\mu u_b \phi_c \bigr]
\bx \sum_{de} k_{de} \bigl[  A'_d B'_e H' + A'_d \phi'_e \del'H'
- \half m_H^2 \phi'_d \phi'_e H' \bigr]
\notag \\
& + \sum_{de} k_{de} \bigl[
u_d B_e^\mu H + u_d \phi_e \del^\mu H \bigr]
\bx f_{abc} \bigl[ F_a'^{\mu\nu} A'_{b\mu} A'_{c\nu}
+ m_{abc}^2 B'_{a\nu} A_b'^\nu \phi'_c \bigr] + [x \otto x'].
\label{eq:higgs-odd} 
\end{align}
By inspection of Table~\ref{tbl:Proca-obstructions}, again with
$c_B = -1$, one sees that the only pairings that contribute are
$\sO_\mu(F,A')$, $\sO_\mu(B,\phi')$ and $\sO_\mu(\phi,B')$. The
resulting structures have only the types $u\phi\phi H$, $uBAH$ and
$u\phi A \del H$.

We compute the coefficients of monomials $\phi_a u_b \phi_c H$,
emerging from~\eqref{eq:higgs-odd} (relabelling the respective
contracted index as~$e$, and adjusting the remaining indices):
\begin{align*}
& \sum_{abc} f_{abc} m_{abc}^2 B_a u_b \phi_c
\bx -\frac{1}{2} m_H^2 \sum_{de} k_{de} \phi'_d \phi'_e H' 
\notag \\
&\quad = m_H^2 \sum_{ebcd} 
f_{ebc} m_{ebc}^2 u_b \phi_c m_e^{-2} k_{de} \phi_d H \,i\dlxx
\equiv \frac{1}{2} m_H^2 \sum_{abc} 
C_{abc} \phi_a u_b \phi_c H \,i\dlxx.
\end{align*}
The monomials $B_a u_b A_c H$ and $\phi_a u_b A_c \del H$ each arise
from three kinds of crossings: the first from $FuA \bx A'B'H'$,
$uBH \bx B'A'\phi'$ and $Bu\phi \bx A'B'H'$; and the second from
$FuA \bx A'\phi' \del'H'$, $Bu\phi \bx A'\phi' \del'H'$ and
$u\phi \del H \bx B'A'\phi'$. These yield (with string deltas
dropped):
\begin{align*}
- \sum_{abc} &\biggl[ \sum_e \bigl( 2f_{ebc} k_{ae}
+ f_{ace} m_{ace}^2 k_{be} m_e^{-2}
- k_{ce} m_e^{-2} f_{abe} m_{abe}^2 \bigr) B_a u_b A_c H
\\
&\quad + \sum_e \bigl( 2f_{ebc} k_{ae}
+ f_{eba} m_{eba}^2 k_{ce} m_e^{-2}
- k_{be} m_e^{-2} f_{eca} m_{eca}^2 \bigr) \phi_a u_b A_c \del H
\biggr] \,i\dlxx.
\end{align*}
Because $f_{a\8e} m_{a\8e}^2 = - f_{e\8a} m_{e\8a}^2$, the two sums
over $e$ are actually identical. Moreover, replacing the $2$ in the
first summands by $2 = (m^2_{ebc} + m^2_{ecb}) m_e^{-2}$, one easily
sees that the sums over $e$ equal $C_{abc} - C_{acb}$. This proves
\eqref{eq:O-abcH}.

Inspection of the three field structures in \eqref{eq:O-abcH} reveals
that only the symmetrized sum over the first of them is resolvable,
namely it is a string variation proportional to
$\dlc[\phi_a \phi_b \phi_c H]$. This proves condition~\eqref{eq:abcH}
for string independence. Formula \eqref{eq:L2-5} follows at once.
\end{proof}

We shall soon find, in Proposition~\ref{pr:ell-ell'}, that $[k_{ab}]$
if nondegenerate%
\footnote{For an example with degenerate $[k_{ab}]$, recall
Remark~\ref{rk:fZ0}. Still, $[k_{ab}]$ is diagonal in that example, too.
On the other hand, the analysis in Prop.~\ref{pr:ell-ell'} shows that
in general $[m_a m_b k_{ab}]$ can only be expected to be a multiple of
a projection matrix. The analogous argument in the local approach
\cite{Aurora,Scharf16} in favour of diagonality is flawed.}
must be diagonal, and $C_{abc} = 0$. That has the following important
consequence.

\begin{prop} 
\label{pr:k-m2}
Assume that $k_{ab} = k_a\,\dl_{ab}$ is diagonal. Then $C_{abc} = 0$
implies
\begin{equation}
\frac{k_a}{m_a^2} = \frac{k_c}{m_c^2} =: K
\label{eq:k-m2} 
\end{equation}
whenever there is a field $c$ such that $f_{abc} \neq 0$.
\end{prop}

\begin{proof}
With $k_{ab}$ diagonal, \eqref{eq:C-abc} reduces to
$$
C_{abc} = \frac{k_a}{m_a^2} m_{abc}^2 f_{abc} + [a\otto c]
= \biggl( \frac{k_a}{m_a^2} - \frac{k_c}{m_c^2} \biggr)
m_{abc}^2 f_{abc} \,.
$$
Thus, Eq.~\eqref{eq:k-m2} follows unless $m_a^2 + m_c^2 = m_b^2$ by
chance. But if so, the relations $m_a^2 + m_b^2 = m_c^2$ and
$m_b^2 + m_c^2 = m_a^2$ cannot both hold, whereby
$k_a m_a^{-2} = k_b m_b^{-2}$ or $k_b m_b^{-2} = k_c m_c^{-2}$, which
at any rate implies that $k_i = K/m_i^2$.
\end{proof}

\begin{remk} 
\label{rk:k-m2}
\begin{enumerate}[itemsep=0pt]
\item 
The quotient $K$ defined in \eqref{eq:k-m2} is only constant over
fields linked by the structure constants of the Lie algebra~$\g$. If
that Lie algebra is the direct sum of two or more simple nonabelian
Lie algebras, these may have independent constants~$K$. This
possibility will be excluded by Proposition~\ref{pr:ell-ell'}(ii).

\item 
It is desirable to show that \eqref{eq:abcH}, in combination with
other conditions from string independence such as \eqref{eq:sum-rule},
has only diagonal solutions $k_{ab} = k_a\,\dl_{ab}$, so that
Proposition~\ref{pr:k-m2} applies. Unfortunately, we are at present
unable to do~so at second order, although Propositions \ref{pr:abcH}
and~\ref{pr:k-m2}, possibly combined with other constraints
like~\eqref{eq:sum-rule}, point to such a result. It does hold for the
electroweak theory -- consult Section~\ref{sec:electro-weak}.
\end{enumerate}
\end{remk}

\subsection{Resolution of obstructions in sectors $[a][b][H][H]$}
\label{ssc:two-higgs}

We continue with the determination of the induced interaction $L_2$ at
second order. Beyond the pieces $L_2^1$ in~\eqref{eq:four-aces} and
$L_2^2$ in~\eqref{eq:L2-2}, as well as $L_2^5$ in~\eqref{eq:L2-5} that
is to be discarded later, there is another piece $L_2^3$ that we
identify now. This piece will be used at third order to determine the
higgs self-couplings.

\begin{prop} 
\label{pr:aaHH}
String independence in the sectors of the form $[a][b][H][H]$ requires
the induced interaction:
\begin{equation}
L_2^3 = \sum_{ab} (3\ell k_{ab} + m_H^2 k^*_{ab}) \phi_a \phi_b H^2,
\label{eq:aaHH} 
\end{equation}
where $k^*_{ab} := \sum_c k_{ac} m_c^{-2} k_{cb}$. 
\end{prop}

\begin{proof}
The obstruction in the sectors of the form $[a][b][H][H]$ arises from
the terms $Q_{1,\higgs} \bx L'_{1,\higgs}$
in~\eqref{eq:L-Q-first-order}. We compute:
\begin{align*}
\sum_{ab} k_{ab} u_a \phi_b \sO_\mu(\del^\mu H, H') 3\ell H'^2
&= 3\ell \sum_{ab} k_{ab} u_a \phi_b H^2 \,i\dlxx,
\\
- m_H^2 \sum_{abc}
k_{ac} u_a H \sO_\mu(B^\mu_c, \phi'_c) k_{cb} \phi'_b H'
&= m_H^2 \sum_{abc} k_{ac} m_c^{-2} k_{cb} u_a \phi_b H^2 \,i\dlxx.
\end{align*}
These are the only nonzero crossings (always assuming that 
$c_B = -1$) without string deltas that have not yet been accounted 
for. Adding $[x \otto x']$, one gets a string variation:
$$
\sO^{(2)} \bigr|_{[a][b][H][H]}
= \dlc\bigl[ \tsum_{ab} (3\ell k_{ab} + m_H^2 k^*_{ab})
\phi_a \phi_b H^2 \bigr] \,i\dlxx,
$$
giving rise to the induced interaction term \eqref{eq:aaHH}.
\end{proof}

\subsection{Wrapping up: the induced interactions at second order}
\label{ssc:wrap-up-second}

We are still free to add to $L_2$ a quartic higgs self-coupling
$L_2^4 := \ell' H^4$. This piece will be needed at third order, where
the coefficient $\ell$ of~$H^3$ in \eqref{eq:L-Q-first-order} and the
new $\ell'$ of~$H^4$ are determined in Proposition~\ref{pr:L3}.

We now collect the complete second-order interaction:
\begin{align}
L_2 &\equiv L_2^1 + L_2^2 + L_2^3 + L_2^4 + L_2^5
\notag \\ 
&= -2(1 + c_F) \sum_{abcde} f_{abe} f_{cde} (A_a A_c) (A_b A_d)
- \frac{1}{4} m_H^2 \biggl( \sum_{ab} k_{ab }\phi_a \phi_b \biggr)^2
\notag \\ 
&\qquad +  \sum_{ab} (3\ell k_{ab} + m_H^2 k^*_{ab}) \phi_a \phi_b H^2
+ \ell' H^4
+ \frac{1}{3} \sum_{abc} C_{abc} m_H^2 \phi_a \phi_b \phi_c H.
\label{eq:L2-complete} 
\end{align}
Remember that these terms were computed under the assumption 
$c_H = c_B = -1$. Additional terms appear for general values of $c_H$
and~$c_B$; those are outlined in Appendix~\ref{aps:cH-cB}.

\begin{remk} 
\label{rk:cH-C}
We shall see in Proposition~\ref{pr:ell-ell'} that $L_2^5$ leads to
non-resolvable obstructions at third order. Therefore $C_{abc} = 0$ is
forced, and then $L_2^5$ is discarded. (For good measure, we already
saw in Sect.~\ref{sec:electro-weak} that $C_{abc} = 0$ in the
electroweak theory.)
\end{remk}

Next, $Q_2|_{I\dl} := 2 \sum_a u_{2a} \,\del L_1/\del A_a$ is given by
Lemma~\ref{lm:uu2}, while $Q_2|_\dl$ is retrieved from the total
derivative in \eqref{eq:minor-crossing}:
\begin{subequations}
\label{eq:Q2} 
\begin{align}
Q_2^\mu \bigr|_\dl(x,x') &= (1 + c_F) \sum_{abcde} 
8 f_{abe} f_{cde} u_a A_c^\mu (A_b A_d) \,i\dlxx,
\label{eq:Q2-explicit} 
\shortintertext{which equals}
Q_2^\mu \bigr|_\dl &= \sum_a u_a \pd{L_2}{A_{a\mu}} \,i\dlxx.
\label{eq:Q2-L2} 
\end{align}
\end{subequations}
Notice that no other total derivatives appear in the computations of
$L_2^2$, $L_2^3$ or $L_2^5$\,.

\section{The boson sector at third order}
\label{sec:third-order}

At third order we see that the process of inducing higher interactions
terminates, and the key parameters of the previous induced
interactions are fixed. We retain the values $c_H = c_B = -1$, see
Remark~\ref{rk:obs-cHcB}.

\subsection{Cancellation of obstructions at third order}
\label{ssc:third-order}

We are now led to compute and resolve the obstructions
\cite{MundRS23, Rehren24a}:
\begin{align}
\sO^{(3)}(x,x',x'') &:= -3i \gS_{xx'x''}
\bigl( \sO(Q_2; L_1'') + \sO(Q_1; L_2')\,\dl_{x'x''} \bigr)
\notag \\
&\stackrel!= \dlc(L_3)\,\dl_{xx'x''}
- \gS_{xx'x''} \del_\mu Q^\mu_3(x,x',x'').
\label{eq:O3} 
\end{align}
We give a quick schematic derivation of \eqref{eq:O3}. At third order,
$\dlc(S)$ is given by $(ig)^3/6$ times
$$
\iiint \bigl( 3 \T\dlc(L_1) L'_1 L''_1
-3i \dlxx( \T\dlc(L_1) L''_2 + \T\dlc(L_2) L''_1)
- \dl_{xx'x''}\,\dlc(L_3) \bigr),
$$
with some dummy delta functions inserted to represent it as a triple
integral. In order to express it in terms of obstruction maps, we
subtract
$$
\iiint \bigl( 3\,\del \T Q_1 L_1' L_1''
- 3i(\dlxx \del \T Q_1 L_2'' + \del \T Q_2 L_1'') \bigr),
$$
and use $\dlc(L_1) = \del Q_1$ and 
$\dlxx\,\dlc(L_2) = \del Q_2 - i\sO^{(2)}$. This yields
\begin{align}
\iiint & \bigl( 3[\T,\del] Q_1 L_1' L_1'' 
- 3i\dlxx [\T,\del] Q_1 L_2'' - 3i [\T,\del] Q_2 L_1''
\notag \\
&\qquad - 3 \T\sO^{(2)} L_1'' - \dl_{xx'x''}\,\dlc(L_3) \bigr).
\label{eq:O3-in-process} 
\end{align}
Here, the fourth term cancels the first by virtue of the Master Ward
Identity \cite{BoasD02} adapted to sQFT \cite{MundRS23, Rehren24a,
Rehren24b},
$$
[\T,\del] Q_1 L_1' L_1'' = \T \sO(Q_1; L_1') L_1''
+ \T \sO(Q_1; L_1'') L'_1 = \T \sO^{(2)}(x,x') L_1''
$$
after symmetrization. The formula \eqref{eq:O3} is the statement that
the symmetrized integrand vanishes up to another total derivative.

It suffices to resolve only the parts of \eqref{eq:O3} without string
deltas, as already discussed.%
\footnote{%
\label{fn:a-priori}
For the string delta parts, there are substantial a~priori
cancellations. Consult \cite[Lemma~4.3]{Rehren24b}.}

\begin{remk} 
\label{rk:pcb}
By a power-counting argument (obstruction maps preserve the total
scaling dimension), we know that $L_3$ must have scaling dimension
$4$, which is impossible for Wick polynomials of degree~$5$.%
\footnote{We thank the anonymous referee for this argument.}
However, because $u$ has scaling dimension $0$, $\sO^{(3)}|_\dl$ may
contain terms of types $uAAAA$ or $u\phi^k H^{4-k}$. These would not
be resolvable, and we must show that such terms do not arise.
\end{remk}

Thus, we must resolve
\begin{equation}
\gS_{xx'x''} \bigl( Q_2\bigr|_\dl \bx L''_1
+ Q_1 \bx (L_2^{1\prime} + L_2^{2\prime} + L_2^{3\prime}
+ L_2^{4\prime} + L_2^{5\prime}) \,\dl_{x'x''}
\bigr) \bigr|_\dl \,.
\label{eq:O3-delta} 
\end{equation}

\smallskip

We begin with the higgs-odd sectors, which, among other things,
determine $\ell$ and~$\ell'$.

\begin{prop} 
\label{pr:ell-ell'}
In the higgs-odd sectors $[a][b][H][H][H]$ and $[a][b][c][d][H]$,
string independence demands that the following conditions be met.
\begin{enumerate}[topsep=3pt, itemsep=0pt]
\item 
The symmetric tensor $C_{abc}$ in Proposition~\ref{pr:abcH} vanishes,
which entails $L_2^5 = 0$.
\item 
The symmetric matrix of higgs couplings $[k_{ab}]$ is of the form
\begin{equation}
k_{ab} = K m_a m_b P_{ab},
\label{eq:higgs-projector} 
\end{equation}
where the matrix $P$ is a projector: $P^2 = P = P^t$, and $K$ is real.
If $[k_{ab}]$ is nondegenerate, then $P = \one$, hence $[k_{ab}]$ is
diagonal and satisfies \eqref{eq:k-m2}, namely,
$k_{ab} = K m_a^2\,\dl_{ab}$\,.
\item 
The values of the cubic and quartic higgs self-couplings are
determined to be
\begin{equation}
\ell = -\half K m_H^2\,, \qquad \ell' = -\quarter K^2 m_H^2\,.
\label{eq:ell-ell'} 
\end{equation}
\end{enumerate}
\end{prop}

With these conditions satisfied, all obstructions in these sectors
cancel each other. In particular, there are no higgs-odd induced
interactions.

\begin{remk} 
\label{rk:C-abc}
If \eqref{eq:k-m2} holds by virtue of condition~(ii), then $C_{abc}$
is identically zero: see the proof of Proposition~\ref{pr:k-m2}. Its
vanishing by condition~(i) imposes no further constraints.
\end{remk}

\begin{proof}[Proof of Proposition~\ref{pr:ell-ell'}]
Notice that $L_2$ and $Q_2|_\dl$ contain no terms in sectors
$[a][H][H][H]$, since there are no crossings that could produce them
at second order. Indeed, $Q_2|_\dl$ is of type $u AAA$
by~\eqref{eq:Q2-explicit}.

\paragraph{Step 1:} the constraint in sectors $[a][b][H][H][H]$. The
crossings that produce obstructions in these sectors are
$Q_{1,\higgs} \bx L_2^{4\prime}$ with the pairing
$\sO_\mu(\del^\mu H,H')$, and $Q_{1,\higgs} \bx L_2^{3\prime}$ with
the pairing $\sO_\mu(B^\mu,\phi')$; plus the $[x \otto x']$ terms.

We compute:
\begin{align*}
& \sum_{ab} k_{ab} u_a \phi_b \,\sO_\mu(\del^\mu H, H')\, 4\ell' H'^3
= 4\ell' \sum_{ab} k_{ab} u_a \phi_b H^3 \,i\dlxx \,,
\\[\jot]
& \sum_{abc} k_{ac} u_a H \,\sO_\mu(B^\mu_c, \phi'_c)\,
(6\ell k_{cb} + 2 m_H^2 k^*_{cb}) \phi'_b H'^2
\\
&\quad = -i \sum_{abc} k_{ab} m_c^{-2}
(6\ell k_{cb} + 2 m_H^2 k^*_{cb}) u_a \phi_b H^3\,\dlxx \,.
\end{align*}
Adding the $[x \otto x']$ part, the total obstruction in this sector
is:
\begin{equation}
\sO^{(2)} \bigr|_{[a][b][H][H][H]}
= 4 \sum_{ab} (2\ell' k_{ab} - 3\ell k^*_{ab} - m_H^2 k^{**}_{ab})
u_a \phi_b H^3 \,i\dlxx,
\label{eq:aaHHH} 
\end{equation}
where $k^*_{ab} := \sum_c k_{ac} m_c^{-2} k_{cb}$ as before, and
$$
k^{**}_{ab} := \sum_c k_{ac} m_c^{-2} k^*_{cb}
= \sum_{cd} k_{ac} m_c^{-2} k_{cd} m_d^{-2} k_{db} \,.
$$ 
Note in passing that the matrices $[k^*_{ab}]$ and $[k^{**}_{ab}]$ are
symmetric. The quantity \eqref{eq:aaHHH} is the string variation of
terms of the type $\phi \phi HHH$. Because these have dimension~$5$,
they are not power-counting renormalizable. Thus, for the obstruction
to be resolvable, its total coefficient must vanish, i.e.,
$2\ell' k_{ab} - 3\ell k^*_{ab} - m_H^2 k^{**}_{ab} = 0$.

\paragraph{Step 2:} constraints in the sectors $[a][b][c][d][H]$.
These sectors contain contributions from the crossings
$Q_2|_\dl \bx L_{1,\higgs}$, $Q_{1,\higgs} \bx L_2^{2\prime}$, 
$Q_{1,\higgs} \bx L_2^{3\prime}$ and $Q_{1,\self} \bx L_2^{5\prime}$.
In the evaluation of $Q_2|_\dl \bx L_{1,\higgs}$, one must use
Lemma~\ref{lm:delta}: the factor $\dlxx$ included in $Q_2|_\dl(x,x')$
can be pulled out of the $\bx$ operation at the price of a
factor~$\half$. Now, this crossing only produces string deltas, see
Remark~\ref{rk:Omu-Anu}, so we ignore it here.

Once more from Table~\ref{tbl:Proca-obstructions} we get, in turn,
with sums over all indices understood:
\begin{align}
k_{ae} u_a B_e^\mu H
\bx -\quarter m_H^2 
k_{rb} k_{cd} \phi'_r \phi'_b \phi'_c \phi'_d
&= m_H^2 k^*_{ab} k_{cd} u_a \phi_b \phi_c \phi_d H\,i\dlxx,
\label{eq:uphi3H} 
\\ 
k_{ab} u_a \phi_b \del^\mu H
\bx (3\ell k_{cd} +  m_H^2 k^*_{cd}) \phi'_c \phi'_d H'^2
&= (6\ell k_{ab} k_{cd} + 2 m_H^2 k_{ab} k^*_{cd})
u_a \phi_b \phi_c \phi_d H \,i\dlxx,
\notag
\end{align}
as well as
\begin{equation}
f_{cde} m_{cde}^2 B_c u_d \phi_e
\bx \tfrac{1}{3} C_{abr} \phi'_a \phi'_b \phi'_r H'
= 2 f_{cde} m_{cde}^2 m_c^{-2} C_{abc}
u_d \phi_a \phi_b \phi_e H \,i\dlxx.
\label{eq:Cobs} 
\end{equation}
On the right-hand side of~\eqref{eq:Cobs}, we write
$f_{cde} m_{cde}^2 m_c^{-2}
= f_{cde} + f_{cde}(m_e^2 - m_d^2) m_c^{-2}$. The first summand is
skewsymmetric and the second is symmetric under $d \otto e$. Since all
other contributions of type $u \phi\phi\phi H$ in the above list have
coefficients that are symmetric in two pairs of indices, the part
$\sum_c f_{cde} C_{abc} u_d \phi_a \phi_b \phi_e H$ cannot be
cancelled by any other term. And because it is not a total derivative,
every $\sum_c f_{cde} C_{abc}$ must vanish. This forces $C_{abc} = 0$.

Adding $[x \otto x']$ to the remaining terms \eqref{eq:uphi3H}, one
gets
\begin{equation}
\sO^{(2)}\bigr|_{[a][b][c][d][H]}
= 2 \sum_{abcd} \bigl(
6\ell k_{ab} k_{cd} + m_H^2 k^*_{ab} k_{cd} 
+ 2 m_H^2 k_{ab} k_{cd}^* \bigr) u_a \phi_b \phi_c \phi_d H \,i\dlxx.
\label{eq:aaaaH} 
\end{equation}
These coefficients must be zero, along with those of~\eqref{eq:aaHHH}.
With the vanishing of \eqref{eq:aaHHH} and~\eqref{eq:aaaaH}, no
obstructions in the higgs-odd sectors remain.

\paragraph{Step 3:} fixing the higgs couplings. Consider the symmetric
matrix $\mu$ with entries $\mu_{ab} := m_a^{-1} k_{ab} m_b^{-1}$ and
notice that $(\mu^2)_{ab} = m_a^{-1} k^*_{ab} m_b^{-1}$ and
$(\mu^3)_{ab} = m_a^{-1} k^{**}_{ab} m_b^{-1}$. One can then rewrite
the vanishing of the respective coefficient matrices in
\eqref{eq:aaHHH} and~\eqref{eq:aaaaH} as:
\begin{equation}
2\ell' \mu - 3\ell \mu^2 - m_H^2 \mu^3 = 0, \qquad
6\ell \mu \ox \mu + m_H^2 \mu^2 \ox \mu + 2 m_H^2 \mu\ox \mu^2 = 0.
\label{eq:moo-moo} 
\end{equation}
Tensoring the first with $2\mu$ from the left, multiplying the second
by $1 \ox \mu$, and adding the results, one is left with:
\begin{equation}
4\ell' \mu \ox \mu = -m_H^2 \mu^2 \ox \mu^2.
\label{eq:ell'-emerges} 
\end{equation}
This is only possible if $\mu^2$ is a multiple of~$\mu$, hence $\mu$
is a multiple of a projector~$P$, that is to say, $\mu = KP$ for some
real number~$K$. With that, \eqref{eq:ell'-emerges} yields
$\ell' = -\quarter K^2 m_H^2$. Then, with $\mu^2 = K \mu$ because
$P^2 = P$, the second equality in~\eqref{eq:moo-moo} also gives
$\ell = -\half K m_H^2$. 

Now $P^2 = P$ also implies $k^*_{ab} = K\,k_{ab}$ and
$k^{**}_{ab} = K^2\,k_{ab}$. And if $[k_{ab}]$ is nondegenerate, then
necessarily $P = \one$ and
$k_{ab} = m_a \mu_{ab} m_b = K m_a^2\,\dl_{ab}$ makes $[k_{ab}]$
diagonal; and the relation \eqref{eq:k-m2} is satisfied (with the same
constant~$K$).
\end{proof}

We arrive at a happy conclusion.

\begin{prop} 
\label{pr:L3}
There are no obstructions at third order in the higgs-even sectors
either, and hence no induced interaction: $L_3 = 0$. This is true
irrespective of the value of~$c_F$.
\end{prop}

\begin{proof}
With the results of Proposition~\ref{pr:ell-ell'}, inspection of
Tables \ref{tbl:higgs-obstructions} and~\ref{tbl:Proca-obstructions}
shows that only a few other crossings may contribute terms without
string deltas:
\begin{enumerate}[topsep=3pt, itemsep=0pt]
\item 
$Q_{1,\self} \bx L_2^{2\prime} = Q_{1,\self} \bx -\quarter m_H^2
\bigl( \sum_{ab} k_{ab} \phi'_a \phi'_b \bigr)^2$ in sectors
$[a][b][c][d][e]$;
\item 
$Q_{1,\self} \bx L_2^{3\prime} = Q_{1,\self} \bx (3\ell + m_H^2 K)
\bigl( \sum_{ab} k_{ab} \phi'_a \phi'_b \bigr) H'^2$ in sectors
$[a][b][c][H][H]$;
\item 
$Q_{1,\self} \bx L_2^{1\prime}$ in sectors $[a][b][c][d][e]$;
\item 
$Q_2|_\dl \bx L_1^{1\prime\prime}$ in sectors $[a][b][c][d][e]$.
\end{enumerate}

For items (i) and (ii), it is enough to consider the crossing:
\begin{align*}
Q_{1,\self} \bx \sum_{ad} k_{ad} \phi'_a \phi'_d
&= \sum_{abcde} f_{ebc} m_{ebc}^2 B_e u_b \phi_c 
\bx k_{ad} \phi'_a \phi'_d
\\ 
&= - \sum_{abce} 2 f_{ebc} m_{ebc}^2 m_e^{-2} k_{ae} 
u_b \phi_c \phi_a \,i\dlxx
= - \sum_{abc} C_{abc} u_b \phi_a \phi_c \,i\dlxx = 0,
\end{align*}
by Proposition~\ref{pr:ell-ell'}(i). Thus, the crossings in (i)
and~(ii) vanish.

Now we examine item (iii), of the type $FuA + Bu\phi \bx (AAAA)''$,
which comes with a factor $(1 + c_F)$. Here we need only the delta
part of $\sO_\mu(F^{\mu\nu}, A')$ in
Table~\ref{tbl:Proca-obstructions}. To shorten the expressions, one
may write contractions with structure constants symbolically as
commutators:
\begin{equation}
\sum_{abc} i f_{abc} X_a Y_b Z_c 
=: \sum_c [X,Y]_c Z_c = \sum_a X_a [Y,Z]_a \,.
\label{eq:commutators} 
\end{equation}
With this notation, omitting the factor $(1 + c_F)$ in
$L_2^{1\prime}$, we reach:
\begin{align*}
Q_{1,\self}^1 &= -2i \sum_a F_a^{\mu\nu} [u,A_\nu]_a \,,
\\
L_2^{1\prime} &= 2 \sum_a [A'^\ka, A'^\la]_a [A'_\ka, A'_\la]_a
= 2 \sum_a A_a'^\ka [A'^\la, [A'_\ka, A'_\la]]_a \,.
\end{align*}
This yields
\begin{align*}
Q_{1,\self}^1 \bx L_2^{1\prime} 
\eqIdl 4 &\sum_a [u,A_\nu]_a [A'^\la, [A'_\ka, A'_\la]]_a
\eta^{\nu\ka} \,\dlxx
\\
= 4 & \sum_a [[u,A^\ka], A^\la]_a [A_\ka, A_\la]_a \,\dlxx.
\end{align*}
The Jacobi identity
$[[u,A^\ka], A^\la] - [[u,A^\la], A^\ka] = [u, [A^\ka, A^\la]]$
reduces this to
$$
2 \sum_a [u,[A^\ka, A^\la]]_a [A_\ka, A_\la]_a \,\dlxx
= 2 \sum_a u_a \bigl[ [A^\ka, A^\la],
[A_\ka, A_\la] \bigr]_a \,\dlxx = 0.
$$

For item (iv), of type $uAAA \bx (FAA + BA\phi)''$, which also comes
with a factor $(1 + c_F)$, we require two-point obstructions
$\sO_\mu(A_\nu, X')$ without Lorentz contraction or skewsymmetrization
of~$\mu,\nu$. These are not listed in
Table~\ref{tbl:Proca-obstructions}, but since string deltas may be
ignored, we may replace such $\sO_\mu(A_\nu, X')$ by their delta parts
$\half \sO_{[\mu}(A_{\nu]}, X')$ -- see Remark~\ref{rk:Omu-Anu}.
Then, in
$$
\sO_\mu \Bigl( \sum_a [u,A_\nu]_a [A^\mu,A^\nu]_a,
\sum_b F_b'^{\ka\la} [A'_\ka,A'_\la]_b \Bigr),
$$
we must pair with $F'$ each of the three $A$-fields on the left.
Pairing of $A_\nu$ with~$F'$ results in:
\begin{align}
c_F &\sum_a [u,\dl_\mu^{[\ka} \dl_\nu^{\la]}
[A'_\ka,A'_\la]\,i\dlxx ]_a [A^\mu, A^\nu]_a
\notag \\
&= 2c_F \sum_a u_a 
\bigl[ [A^\ka, A^\la], [A_\ka, A_\la] \bigr]_a \,i\dlxx = 0.
\label{eq:first-obstruction-null} 
\end{align}
The pairing of $A^\mu$ with~$F'$ requires $\sO_\mu(A^\mu, F')$, which
has no delta part. The skewsymmetrized pairing of $A^\nu$ with~$F'$
yields
\begin{align}
\frac{1}{2}\,c_F &\sum_a [u, A_\nu]_a \bigl[ A^\mu, 
\eta^{\nu\rho} \dl_{[\mu}^{[\ka} \dl_{\rho]}^{\la]} [A'_\ka, A'_\la]
\,i\dlxx \bigr] 
\notag \\
&= 2c_F \sum_a [[u,A^{[\ka}], A^{\la]}]_a [A'_\ka, A'_\la]_a \,i\dlxx.
\label{eq:second-obstruction} 
\end{align}
Thanks to  the Jacobi identity, this again equals:
\begin{equation}
2 c_F \sum_a [u,[A^\ka, A^\la]]_a [A'_\ka, A'_\la]_a \,i\dlxx 
= 2 c_F \sum_a u_a \bigl[
[A'^\ka, A'^\la], [A'_\ka, A'_\la] \bigr]_a \,i\dlxx = 0.
\label{eq:second-obstruction-null} 
\end{equation}
Thus all obstructions vanish identically, irrespective of the value
of~$c_F$ in (iii) and~(iv).
\end{proof}

\begin{remk} 
\label{rk:termination}
The termination of induced couplings after the second order is a
\textit{necessary} feature, because higher-order induced interactions
would not be renormalizable.
\end{remk}

\subsection{Higher orders $n > 3$}
\label{ssc:Icarus}

Higher-order interactions $L_n$ are determined by the parts of
$\sO^{(n)}$ without string deltas. At fourth order,
$\sO(Q_1; L_3)|_\dl = 0$ because $L_3 = 0$;
$\sO(Q_3; L_1''')|_\dl = 0$ because $Q_3|_\dl = 0$; and
$\sO(Q_2|_\dl; L_2'') = 0$ since $Q_2|_\dl$ and $L_2$ are Wick
polynomials in $u$ and~$A$, with $\sO_\mu(A^\mu, A')|_\dl = 0$.

This argument generalizes to all orders by induction, see also
Remark~\ref{rk:pcb}. The only open question is whether all
obstructions with string deltas are total derivatives, i.e., the
existence of $Q_n|_{I\dl}$.

\section{Outlook}
\label{sec:outlook}

We have studied interactions between particles of spin and helicity
$1$ and scalar particles on the string-localized Hilbert-space fields
provided by sQFT. Given the particle content of the electroweak
theory, the condition of string independence (SI) of the $\bS$-matrix
fixes all coupling coefficients, up to a freedom of renormalization,
see Remark~\ref{rk:cF}, and predicts precisely the known interactions
of the Standard Model.

We have also laid the grounds for an SI analysis of more general
models of massive and massless vector bosons. Resolution of
obstructions to~SI in the general case consists of a plethora of
polynomial constraints on coupling coefficients and masses. Such a
general solution may be quite difficult to characterize. It might be
interesting to know whether GUT models with SSB satisfy all
consistency constraints of sQFT.

\smallskip

For the models with one scalar particle (one higgs) studied in this
paper, we may define skewsymmetric matrices
\begin{equation}
(\ga^a)_{cd} :=  \frac{m_{cad}^2}{2 m_c m_d}\, f_{cad}
\equiv \frac{m_{cad}^2}{2 m_c m_d}\, \ad(\xi^a)_{cd} \,,
\label{eq:gamma} 
\end{equation}
whose indices run over the ``massive'' particles $c,d$ only. For
massless $a$, $m_{cad}^2 = 2m_c m_d$ holds, so $(\ga^a)_{cd}$ equals
$f_{cad} = \ad(\xi^a)_{cd}$, the adjoint representers of the
``massless'' Lie subalgebra $\gh$ that organizes the massless
particles into multiplets (representations of $\gh$), according to
Proposition~\ref{pr:boson-sector}.

When the Lie algebra structure constants and the higgs coupling
coefficients $k_{ab}$ are expressed in terms of the matrices $\ga^a$
and the matrix projector $P$ of Eq.~\eqref{eq:higgs-projector}, all
conditions for string independence, namely conditions (i) and~(ii) in
Proposition~\ref{pr:ell-ell'} together with Eq.~\eqref{eq:ff-kk}, can
be displayed as a system of matrix equations:
\begin{gather} 
P^2 = P = P^t,  \qquad  [P, \ga^a] = 0 \quad \text{(all $a$)},
\notag \\
[\ga^a, \ga^b] = \sum_{e:m_e=m_b} f_{abe} \ga^e \quad
\text{($a$ massless, $b$ massive)}, 
\label{eq:multiplets} 
\\
[\ga^a, \ga^b] - \sum_e f_{abe} \ga^e 
= \quarter m_a m_b K^2\, P^a \w P^b  \quad
\text{($a,b$ massive)},
\label{eq:deformed-Lie} 
\end{gather}
where $P^a$ are the column vectors of~$P$ and the sums in
\eqref{eq:deformed-Lie} run over all indices $e$, massive or massless.
This rewriting teases out an algebraic structure underlying the SI
conditions that could be of use to analyze more general admissible
mass patterns. In particular, \eqref{eq:multiplets} says that the
adjoint action of $\gh$ on the space of massive $\ga^b$ coincides with
its action on the space of massive $\xi^b$, which splits it into
representations of the Lie algebra~$\gh$. By \eqref{eq:deformed-Lie},
the higgs couplings compensate for the failure of the
``mass-deformed'' massive generators $\ga^a$ to satisfy the Lie
algebra of the~$\xi^a$.

\begin{remk} 
\label{rk:sum-rule}
On dividing the sum rule \eqref{eq:sum-rule} by $m_a^2 m_b m_c$, the
right-hand side becomes $K^2 (P_{aa} P_{bc} - P_{ab} P_{ac})$. Summed
over~$a$, this is $K^2(r - 1) P_{bc}$, where $r$ is the rank of the
projector~$P$. Thereby, \eqref{eq:sum-rule} gives an explicit formula
for $P_{ab}$ in terms of the masses and the structure constants of the
Lie algebra. The idempotent property $P^2 = P$ is then a direct
constraint relation (not involving $k_{ab}$) among the latter.
\end{remk}

\appendix

\section{Two-point obstructions}
\label{app:two-pt-obstructions}

This appendix outlines the construction of Tables
\ref{tbl:higgs-obstructions} and~\ref{tbl:Proca-obstructions} in
Section~\ref{sec:in-the-way}.

\subsection{Two-point functions}
\label{aps:two-pointers}

Let $W_m(x - x') = \pq{\vf(x)}{\vf(x')}$ be the two-point function of
a free scalar field of mass~$m$, so that
$(\square + m^2) W_m(x - x') = 0$.

For two-point functions involving derivatives of
fields we apply the rules
$$
\pq{\del^\mu X(x)}{Y(x')} = \del^\mu\bigl( \pq{X}{Y} \bigr)(x - x')
\word{and}
\pq{X(x)}{\del'^\nu Y(x')} = -\del^\nu\bigl( \pq{X}{Y} \bigr)(x - x').
$$
This settles all two-point functions of the higgs field and its
derivatives, in particular:
\begin{equation}
\pq{\del_\mu H(x)}{\del'_\nu H(x')} 
= - \del_\mu \del_\nu W_{m_H}(x - x').
\label{eq:twopt-dH-dH} 
\end{equation}

Turning to the fields in the Proca sector, recall from
Sect.~\ref{ssc:notations} that
$A_\mu(x,c) = I^\nu_c F_{\mu\nu}(x)$,
$B_\nu(x) = - m^{-2} \del_\mu F^{\mu\nu}(x)$, 
$\phi(x,c) = I^\mu_c B_\mu(x)$, and $u(x,h) = \dlc(\phi(x,c))$. One
therefore obtains all two-point functions of string-localized fields
from
\begin{equation}
\pq{F_{\mu\nu}(x)}{F_{\ka\la}(x')} 
= (\eta_{\mu\ka} \del_\nu\del_\la - \eta_{\nu\ka} \del_\mu\del_\la 
+ \eta_{\nu\la} \del_\mu\del_\ka - \eta_{\mu\la} \del_\nu\del_\ka)
W_m(x - x'),
\label{eq:twopt-F-F} 
\end{equation}
by applying the rules:
$$
\pq{I_c^\mu X(x)}{Y(x')} = I_c^\mu\bigl( \pq{X}{Y} \bigr)(x - x')
\word{and}
\pq{X(x)}{I_c^\nu Y(x')} = I'^\nu_c\bigl( \pq{X}{Y} \bigr)(x - x'),
$$
where $I'^\nu_c$ acts on the argument~$x'$. Let us abbreviate
$I^\mu \equiv I_c^\mu$ and $I'^\nu \equiv I'^\nu_c$, as well as
$X \equiv X(x)$ and $X' \equiv X(x')$ for fields. The argument of 
every two-point function is $(x - x')$. Formula
\eqref{eq:our-TFC} now reads $(\del I) = (I \del) = -\id$, and also
$(\del I') = (I' \del) = +\id$. On using the Klein--Gordon equation,
this yields in particular:
\begin{align}
\pq{B_\mu}{B'_\nu} \equiv \pq{B_\mu(x)}{B_\nu(x')}
= -(\eta_{\mu\nu} + m^{-2} \del_\mu \del_\nu) W_m(x - x'),
\label{eq:twopt-B-B} 
\end{align}
The same rules apply in the photon sector, using the two-point
function \eqref{eq:twopt-F-F} with $m = 0$ for the Faraday tensor, and
the definitions $A_\mu(x,c) := I^\nu_c F_{\mu\nu}(x)$ and
$u := -I^\mu_c \dlc(A_\mu)$.

\subsection{Propagators}
\label{aps:good-propagators}

Defining time-ordered vacuum expectation values naïvely with the help
of the Heaviside function $\theta(x^0 - x'^0)$ is in general
illegitimate, since one is multiplying distributions. For
point-localized fields, it is well known that locality and covariance
ensure that the naïve definition is well defined outside the
``diagonal'' set $x = x'$. Therefore, one needs to extend that naïve
definition to the points $x = x'$.

One extension is given by the so-called ``kinematic'' propagator,
which amounts to replacing $W_m$ by $i D^F_m$, where $D^F_m$ denotes
the Feynman propagator \eqref{eq:Feynman-propagator} of a scalar field
of mass~$m$. However, the extension is in general not unique: one may
add (derivatives of) $\dl(x - x')$ with the correct symmetry and
Lorentz transformation behaviour. This ``renormalization'' of
propagators is constrained by the condition that it must not exceed
the scaling dimension of the kinematic propagator.

For string-localized fields, regarded as distributions in $x$ and~$e$,
the situation looks far more delicate because the ``string diagonal''
consists of all points $x + se = x' + s'e'$ ($s,s' \geq 0$). However,
Ga{\ss} showed in \cite[Thm.~4.5]{Gass22b} that when the
string-localized fields are smeared with~$c(e)$ and regarded as
distributions in $x$~only, the relevant diagonal is again $x = x'$. In
particular, this rules out nontrivial commutation between the
operations of time-ordering $\T$ and string variation $\dlc$, which in
principle should be taken into account, since obstructions of this
sort vanish after smearing with $c(e)$. Therefore the allowed
renormalizations are still just of the type $\dl(x - x')$ and its
derivatives, occurring only when the scaling dimension is~$\geq 4$.

In the current context, since string integrations lower the scaling
dimension, only the propagators of local fields with scaling
dimension~$2$ admit in principle such renormalizations. This pertains
to the time-ordering of \eqref{eq:twopt-dH-dH}, \eqref{eq:twopt-F-F},
and \eqref{eq:twopt-B-B}:
\begin{align}
\Tpq{\del_\mu H(x)}{\del'_\nu H(x')} 
&= -i \del_\mu \del_\nu \,D^F_{m_H}(x - x')
+ i c_H \eta_{\mu\nu} \,\dl(x - x'),
\notag \\                     
\Tpq{F_{\mu\nu}(x)}{F_{\ka\la}(x')}
&= -i \del_{[\mu} \eta_{\nu][\ka} \del_{\la]} \,D^F_m(x - x')
- i c_F \eta_{\mu[\ka} \eta_{\la]\nu} \,\dl(x - x'),
\notag \\                    
\Tpq{B_\mu(x)}{B_\nu(x')} 
&= -i(\eta_{\mu\nu} + m^{-2} \del_\mu \del_\nu) D^F_m(x - x')
+ i c_B m^{-2} \eta_{\mu\nu} \,\dl(x - x').                      
\label{eq:propagator-list} 
\end{align}
The real coefficients $c_H, c_F, c_B$ are free parameters at this
point. All other propagators are ``kinematic'', that is, they are
given by replacing $W_m$ by $iD^F_m$.

\subsection{Computing two-point obstructions}
\label{aps:nuisances}

We determine here the two-point obstructions \eqref{eq:2pt-obst-def}
for all relevant fields. In the first term of~\eqref{eq:2pt-obst-def},
with the derivative inside the $\T$-product, one uses the equations of
motion \eqref{eq:eom} for the fields and computes the resulting
propagators as in the previous subsection. One then subtracts the
derivatives of the propagators in the second term. The Green function
property \eqref{eq:Feynman-propagator} of the Feynman propagators
produces delta functions, added to the deltas appearing in the
renormalized propagators (when applicable).

An example may suffice to illustrate the general procedure. From
$$
\pq{B^\mu(x)}{\phi(x')} 
= I'_\ka \bigl( -\eta^{\mu\ka} - m^{-2} \del^\mu\del^\ka \bigr)
W_m(x - x')
= -(I'^\mu + m^{-2} \del^\mu) W_m(x - x')
$$
using $(I'\del) = \id$, we conclude
$\Tpq{B^\mu(x)}{\phi(x')} = -(I'^\mu + m^{-2}\del^\mu)iD^F_m(x - x')$.
Thus
\begin{align*}
\sO_\mu(B^\mu,\phi') &:= \Tpq{\del_\mu B^\mu(x)}{\phi(x')}
- \del_\mu\Tpq{B^\mu(x)}{\phi(x')}
\\
&= \del_\mu(I'^\mu + m^{-2}\del^\mu) i D^F_m(x - x')
= m^{-2}(\square + m^2) iD^F_m(x - x')
\\
&= -m^{-2} \,i\dl(x - x'). 
\end{align*}

This results in the Tables \ref{tbl:higgs-obstructions}
and~\ref{tbl:Proca-obstructions} of two-point obstructions. The last
line of Table~\ref{tbl:Proca-obstructions} is obtained by string
variation $\dlc$ of the line before~it. All entries also pertain to
$\sO_\mu(X(x,c), Y(x',c'))$ with $c'$ independent of~$c$.

\begin{remk} 
\label{rk:Omu-Anu}
Table~\ref{tbl:Proca-obstructions} displays only the
Lorentz-contracted and skewsymmetrized parts of the two-point
obstructions $O_\mu(A_\nu, X')$ that are needed at second order. The
traceless symmetric part is not obtained with this approach, because
one cannot use the equations of motion for propagators
$\Tpq{\del_\mu A_\nu}{X'}$. Fortunately, those are not needed at
second order; and at third order only the part of $\sO_\mu(A_\nu, X')$
without string deltas is required. This is easily found: because
$\Tpq{\del_\mu A_\nu}{X'}$ for $X' = A',B',\phi'$ have respective
scaling dimensions $3,3,2$, the corresponding $\sO_\mu(A_\nu, X')$
cannot include delta parts (having dimension at least~$4$). For
$X' = F'$, the propagator $\Tpq{\del_\mu A_\nu}{F'}$ of dimension~$4$
does admit a delta part of $\sO_\mu(A_\nu, F'^{\ka\la})$. Now, for
Lorentz-symmetry reasons, it must be skewsymmetric in $\mu \otto \nu$,
and therefore it equals $\half \sO_{[\mu}(A_{\nu]}, X')$ in
Table~\ref{tbl:Proca-obstructions}. The suppressed traceless symmetric
parts of $O_\mu(A_\nu, X')$ are purely string deltas.
\end{remk}

\section{Some details of the second-order resolutions}
\label{app:some-lemmata}

Here we give a few lemmas that complete the determination of $L_2$
and~$Q_2$ in the resolution \eqref{eq:seconds-out} at second order.

\subsection{Disposing of string deltas}
\label{aps:string-delta-anathema}

We show that all second-order obstructions involving string deltas are
automatically derivatives, contributing to $Q_2|_{I\dl}$. See
Remark~\ref{rk:no-string-deltas}.

First comes a preparatory observation.

\begin{lemma} 
\label{lm:L1}
If $\dlc L_1(c) = \del Q_1$ where $L_1$ is a Wick polynomial in the
fields $A_a$, $\phi_a$ and string-independent fields, and $Q_1$ is
linear in~$u_a$, then:
\begin{equation}
\mathrm{(i)}\quad
\pd{L_1}{A_{a\mu}} = \pd{Q^\mu_1}{u_a},  \qquad
\mathrm{(ii)}\quad
\pd{L_1}{\phi_a} = \del_\mu \biggl( \pd{Q^\mu_1}{u_a} \biggr), \qquad
\mathrm{(iii)}\quad
\pd{L_1}{\phi_a} = \del_\mu \biggl( \pd{L_1}{A_{a\mu}} \biggr).
\label{eq:L1} 
\end{equation}
In particular, when $L_1$ does not contain $\phi_a$, the quantity
$\del L_1/\del A_{a\mu}$ is conserved.
\end{lemma}

The latter case applies for the photon field, where $\phi_a$ does not
exist.

\begin{proof}
The comparison of
$$
\dlc L_1 
= \sum_a \pd{L_1}{\phi_a}\, u_a + \pd{L_1}{A_{a\mu}}\, \del_\mu u_a
$$
with
$$
\del_\mu Q_1^\mu 
= \sum_a \del_\mu \biggl( \pd{Q_1^\mu}{u_a}\, u_a \biggr)
= \sum_a \del_\mu \biggl (\pd{Q_1^\mu}{u_a} \biggr) u_a
+ \pd{Q_1^\mu}{u_a}\, \del_\mu u_a
$$
immediately yields (i) and~(ii). Formula~(iii) and the last statement
are obvious consequences of (i) and~(ii).
\end{proof}

\begin{lemma} 
\label{lm:uu2}
For the interactions $S_1 = L_1$ and $Q_1^\mu$ as specified in
Eq.~\eqref{eq:L-Q-first-order}, all second-order obstructions
involving string deltas are total derivatives. They determine the part
$Q_2^\mu\bigr|_{I\dl}(x,x')$ of $Q_2^\mu$ to arise by a simple
replacement of $u(x)$ by $2 u_2(x,x')$ in $Q_1^\mu$:
\begin{align}
Q_2^\mu\bigr|_{I\dl}(x,x') 
&= 2 \sum_a \pd{Q^\mu_1}{u_a}(x)\, u_{2a}(x,x'),
\label{eq:uu2} 
\\
\shortintertext{where}
u_{2a}(x,x') 
&:= - \sum_{bc} f_{abc}\, u_b(x') A_{c\nu}(x') \,I^\nu \dlxx.
\label{eq:uu2-bis} 
\end{align}
\end{lemma}

\begin{proof}
Because $Q_1^\mu$ contains the fields $A_\nu$ only in the
skewsymmetric combination $F^{\mu\nu} A_\nu$, the third line of
Table~\ref{tbl:Proca-obstructions} does not contribute to
$\sO(S_1; S_1')$. Therefore, the string deltas may only arise through
$\sO_\mu(F^{\mu\nu}, A')$ and $\sO_\mu(F^{\mu\nu}, \phi')$. They
contribute
\begin{align*}
& \pd{Q_1^\mu}{F_a^{\mu\nu}} \biggl[
-i I'_\nu \dlxx \pd{L_1'}{\phi'_a}
-i (\del'^\ka I'_\nu \dlxx)\, \pd{L_1'}{A'^\ka_a} \biggr]
\\
&\quad = \del'^\ka \biggl[ -i \pd{Q_1^\mu}{F_a^{\mu\nu}}\,I'_\nu \dlxx
\,\pd{L_1'}{A'^\ka_a} \biggr]
= \del'^\ka \biggl[ -i \pd{Q_1^\mu}{F_a^{\mu\nu}}\, I'_\nu \dlxx
\,\pd{Q_1'^\ka}{u'_a} \biggr]
\end{align*}
by Lemma~\ref{lm:L1}(iii) and~(i). 

Now \eqref{eq:uu2} follows from the formula
\eqref{eq:L-Q-first-order-Q1} and the condition \eqref{eq:seconds-out}
for resolving second order obstructions that determines
$Q_2^\mu(x,x')$.
\end{proof}

\subsection{The case of general $c_H$ and $c_B$}
\label{aps:cH-cB}

In the main body of the paper, we have computed second-order
obstructions with the choice of renormalization constants
$c_H = c_B = -1$. The next Lemma shows that the additional
contributions for general $c_H$, $c_H = c_B = -1$. The additional
contributions for general $c_H$, $c_B$ are always resolvable and have
a rather simple form. The result reflects the circumstance that a
delta function in a propagator amounts to the contraction to a new
quartic vertex of two cubic vertices connected by that propagator.

\begin{lemma} 
\label{lm:cH-cB}
The additional second-order obstruction when $c_H$ and $c_B$ differ
from the distinguished choice $c_H = c_B = -1$ is
$$
\sO^{(2)*}(x,x') = \dlc[L_2^*] - \gS_{xx'} \del_\mu Q_2^{*\mu}
$$
with
\begin{align}
L_2^* &= (1 + c_H) \pd{L_1}{(\del_\ka H)}\,\pd{L_1}{(\del^\ka H)}
+ (1 + c_B) \sum_e m_e^{-2} \pd{L_1}{B_{e\ka}}\,\pd{L_1}{B_e^\ka}\,,
\notag \\
Q_2^{*\mu} &= 2(1 + c_H) 
\pd{Q_1^\mu}{(\del_\ka H)}\,\pd{L_1}{(\del^\ka H)}\,i\dlxx
+ 2 (1 + c_B) \sum_e \pd{Q_1^\mu}{B_{e\ka}}\,\pd{L_1}{B_e^\ka}
\,i\dlxx.
\label{eq:obs-cH-cB} 
\end{align}
The relation \eqref{eq:Q2-L2} also holds for $Q_2^*$ and $L_2^*$, 
namely 
$\dsp Q^*_2\bigr|_\dl = \sum_a u_a \pd{L_2^*}{A_{a\mu}}\, i\dlxx$.
\end{lemma}

\begin{remk} 
\label{rk:obs-cHcB}
$L_2^*$ has terms of type $AA\phi\phi$, $AAHH$ and $AA\phi H$ (where
the last of these also comes with coefficients $C_{abc}$). We expect,
from experience with a simpler model with all masses equal
\cite{Rehren24b}, that the string delta parts of the third-order
obstructions (which we are not considering here, but must be
separately derivatives because they cannot be part of $\dlc(L_3)$) put
further constraints on the renormalization constants, leaving only the
choice $c_H = c_B = -1$. The rationale is similar to that
of~\cite{MundRS23}, where that choice was motivated by the complete
absence of string deltas. This shortcut is not possible in nonabelian
models, because of Lemma~\ref{lm:uu2}.
\end{remk}

\begin{proof}[Proof of Lemma~\ref{lm:cH-cB}]
Notice that $\del Q_1^\mu/\del(\del_\al H)$ and
$\del Q_1^\mu/\del B_\al$ both contain a factor $\eta^{\mu\al}$. It is
convenient to write:
$$
\pd{Q_1^\mu}{(\del_\al H)} =: \eta^{\mu\al} \pd{Q_1}{(\del H)}
\word{and}
\pd{Q_1^\mu}{(B_{e\al})} =: \eta^{\mu\al} \pd{Q_1}{(\del B_e)}\,.
$$
Similarly, one can abbreviate
$$
\pd{Q_1^\mu}{F_{e\al\bt}} 
= \frac{1}{2} \sum_{ab} \eta^{\mu\al} f_{eab}\,(2 u_a A_b^\bt)
- [\al \otto \bt] 
=: \frac{1}{2} \eta^{\mu\al} \Bigl( \pd{Q_1}{(\del F)} \Bigr)^\bt 
- [\al \otto \bt].
$$
After inspection of the $c_H$- and $c_B$-dependent entries in Tables
\ref{tbl:higgs-obstructions} and~\ref{tbl:Proca-obstructions}, one
must compute:
\begin{align}
\frac{\del}{\del c_H}\, Q_1 \bx L_1'
&= \biggl[ \pd{Q^\mu_1}{H} \,i\dlxx 
- \pd{Q_1}{(\del H)} \,i\del^\mu\dlxx \biggr]
\pd{L_1'}{(\del'^\mu H')}
\notag \\
&= \biggl[ \pd{Q^\mu_1}{H}
+ \del^\mu \Bigl( \pd{Q_1}{(\del H)} \Bigr) \biggr]
\pd{L_1'}{(\del'^\mu H')} \,i\dlxx
- \del^\mu \biggl(
\pd{Q_1}{(\del H)}\,\pd{L_1'}{(\del'^\mu H')} \,i\dlxx \biggr)
\label{eq:OL-cH-cB-extra} 
\\
\shortintertext{and}
\frac{\del}{\del c_B}\, Q_1 \bx L_1' 
&= \sum_e \biggl[ -\Bigl( \pd{Q_1}{F_e} \Bigr)^\mu i\dlxx
- m_e^{-2} \pd{Q^\mu_1}{\phi_e} \,i\dlxx
- m_e^{-2} \pd{Q_1}{B_e} \,i\del^\mu\dlxx \biggr] \pd{L_1'}{B_e'^\mu}
\notag \\
&= \sum_e m_e^{-2} \biggl[ -m_e^2 \Bigl( \pd{Q_1}{F_e} \Bigr)^\mu
- \pd{Q^\mu_1}{\phi_e} + \del^\mu \Bigl( \pd{Q_1}{B_e} \Bigr) \biggr]
\pd{L_1'}{B_e'^\mu} \,i\dlxx 
\notag \\
&\qquad - \del^\mu \biggl( \sum_e m_e^{-2} 
\pd{Q_1}{B_e}\,\pd{L_1'}{B_e'^\mu} \,i\dlxx \biggr).
\label{eq:OL-cH-cB} 
\end{align}

Next, there are the remarkable relations:
\begin{align}
\pd{Q^\mu_1}{H} + \del^\mu \biggl( \pd{Q_1}{(\del H)} \biggr)
&= \dlc \biggl[ \pd{L_1}{(\del_\mu H)} \biggr],
\notag \\
\biggl( -m_e^2 \pd{Q_1}{F_e} \biggr)^\mu - \pd{Q^\mu_1}{\phi_e} 
+ \del^\mu \biggl( \pd{Q_1}{B_e} \biggr) 
&= \dlc \biggl[ \pd{L_1}{B_{e\mu}} \biggr].
\label{eq:Q-dlc-L} 
\end{align}

These are verified by direct computation. For the first:
$$
\sum_{ab} k_{ab} \bigl( B_a^\mu u_b + \del^\mu (\phi_a u_b) \bigr)
= \sum_{ab} k_{ab} \bigl( A_a^\mu u_b + \del^\mu u_a \phi_b \bigr)
= \dlc \biggl[ \sum_{ab} k_{ab} A_a^\mu \phi_b \biggr]
= \dlc\biggl[ \pd{L_1}{(\del_\mu H)} \biggr],
$$
and the second for $Q_{1,\higgs}$:
$$
\sum_b k_{eb} \bigl( - u_b \del^\mu H + \del^\mu(u_b H) \bigr)
= \sum_b k_{eb}\, \del^\mu u_b H 
= \dlc\biggl[ \sum_b k_{eb} A_b^\mu H \biggr] 
= \dlc\biggl[ \pd{L_{1,\higgs}}{B_{e\mu}} \biggr],
$$
and for $Q_{1,\self}$:
\begin{align*}
& \sum_{bc} -2 m_e^2 f_{ebc} u_b A_c^\mu 
- f_{bce} m_{bce}^2 B_b^\mu u_c
+ f_{ebc} m_{ebc}^2 \del^\mu(u_b \phi_c)
\\
&\quad = \sum_{bc} -(m_{ebc}^2 + m_{ecb}^2) f_{ebc} u_b A_c^\mu
+ f_{ebc} m_{ebc}^2 B_c^\mu u_b
+ f_{ebc} m_{ebc}^2 (\del^\mu u_b \phi_c + u_b \del^\mu \phi_c)
\\
&\quad = \sum_{bc} f_{ebc} m_{ebc}^2
\bigl( A_b^\mu u_c + \del^\mu u_b \phi_c \bigr)
= \sum_{bc} f_{ebc} m_{ebc}^2 \,\dlc[A_b^\mu \phi_c]
= \dlc\biggl[ \pd{L_{1,\self}}{B_{e\mu}} \biggr].
\end{align*}

When the relations \eqref{eq:Q-dlc-L} are inserted into
\eqref{eq:OL-cH-cB}, and $[x\otto x']$ is added, the formulas
\eqref{eq:obs-cH-cB} follow from Eq.~\eqref{eq:seconds-out}. The final
statement is a consequence of the relations:
$$
\pd{Q^\mu_1}{(\del^\ka H)} = \sum_b u_b \frac{\del}{\del A_{b\mu}} 
\biggl( \pd{L_1}{(\del^\ka H)} \biggr) \word{and}
\pd{Q^\mu_1}{B_e^\ka} = \sum_b u_b \frac{\del}{\del A_{b\mu}} 
\biggl( \pd{L_1}{B_e^\ka} \biggr)
$$
which hold by inspection.
\end{proof}

\subsection{Delta functions within $Q_2^\mu$}
\label{aps:doubling-down}

The resolution of second-order obstructions of the form 
$\del_\mu Y^\mu(x)\,\dlxx$ and the computation of third-order
obstructions of the form $Q_2(x,x') \bx L(x'')$, where
$Q_2(x,x') = Y(x)\,\dlxx$, bring in some unexpected factors of~$2$.

\begin{lemma} 
\label{lm:delta}
For $Q^\mu(x,x')$ of the form $Q^\mu(x,x') = Y^\mu(x)\,\dlxx$, the
following relations hold:
\begin{align*}
\mathrm{(i)} \quad
& 2\gS_{xx'} \bigl( \del_\mu Q^\mu(x,x') \bigr)
= \del_\mu Y^\mu(x)\,\dlxx,
\\[\jot]
\mathrm{(ii)} \quad
& 2 \gS_{xx'}(Y(x)\,\dlxx) \bx X(x'')
= \bigl( Y(x) \bx X(x'') \bigr) \,\dlxx.
\end{align*}
\end{lemma}

\begin{proof}
From $(\del + \del')\dlxx = 0$ it follows that
\begin{equation}
(\del_x + \del_x')(f(x)\,\dlxx) = \dlxx\,(\del_x + \del_x')f(x)
=  \del f(x) \,\dlxx\,.
\label{eq:del-xx'} 
\end{equation}
The left-hand side of~(i) is
$$
2 \gS_{xx'} (\del_x Q)(x,x') = (\del_x Q)(x,x') + [x \otto x']
= (\del_x + \del_{x'}) (Y(x)\,\dlxx) = (\del Y(x))\,\dlxx.
$$
The left-hand side of~(ii) is
\begin{align*}
& Y(x) \,\dlxx \bx X(x'') + [x \otto x']
\\
&\quad = \T\bigl[ (\del_x + \del_{x'})(Y(x)\,\dlxx) X(x'') \bigr]
- (\del_x + \del_{x'}) \T\bigl[ Y(x)\,\dlxx\, X(x'') \bigr].
\end{align*}
Using \eqref{eq:del-xx'} again, one gets
\begin{align*}
& \T\bigl[ (\del Y(x))\,\dlxx\, X(x'') \bigr]
- (\del_x + \del_x')\bigl( \dlxx \T[Y(x) X(x'')] \bigr)
\\[\jot]
&\quad = \bigl( \T[\del Y(x) X(x'')]
- \del \T[Y(x) X(x'')] \bigr)\,\dlxx.
\tag*{\qed}
\end{align*}
\hideqed
\end{proof}


\subsection*{Acknowledgments}

We are indebted to Jens Mund and Bert Schroer for their contributions
to early stages of this work. JMG-B and JCV gladly acknowledge the
good offices of the Centro de Investigación en Matemática Pura y
Aplicada via a research project at the University of Costa Rica. JCV
thanks Christian Ga{\ss} for helpful discussions during a visit to
Warsaw. This research was partially supported by the University of
Warsaw Thematic Research Programme ``Quantum Symmetries''.

\bigskip


\end{document}